\documentclass[aps,prd,notitlepage,nofootinbib,superscriptaddress,longbibliography]{revtex4-1}
\usepackage{amsmath,amssymb,amsthm,latexsym, mathtools,calc}
\usepackage[english]{babel}
\usepackage{graphicx,color}
\usepackage{xspace}
\usepackage{pifont,dsfont}
\usepackage{marvosym}
\usepackage{slashed}
\usepackage{subfigure}
\usepackage[dvipsnames]{xcolor}
\usepackage{hyperref}

\newcommand{\der}[3][]{\frac{\mathrm{d}^{#1}{#2}}{{\mathrm{d}{#3}}^{#1}}}

\newcommand{\be}{\begin{equation}}
\newcommand{\ee}{\end{equation}}
\newcommand{\beq}{\begin{equation}}
\newcommand{\eeq}{\end{equation}}
\newcommand{\bes}{\begin{eqnarray}}
\newcommand{\ees}{\end{eqnarray}}
\newcommand{\bqa}{\begin{eqnarray}}
\newcommand{\eqa}{\end{eqnarray}}
\newcommand{\bea}{\begin{eqnarray}}
\newcommand{\eea}{\end{eqnarray}}

\newcommand{\N}{\mathds{N}}

\newcommand{\R}{\mathds{R}}
\newcommand{\C}{\mathds{C}}

\newcommand{\nn}{\nonumber}

\renewcommand{\N}{\mathbb{N}}
\renewcommand{\R}{\mathbb{R}}
\renewcommand{\C}{\mathbb{C}}

\newcommand{\su}{\mathfrak{su}}
\newcommand{\SU}{\mathrm{SU}}
\renewcommand{\u}{\mathfrak{u}}

\newcommand{\Sp}{\mathrm{Sp}}
\newcommand{\Var}{\mathrm{Var}}
\newcommand{\Cov}{\mathrm{Cov}}
\newcommand{\rank}{\mathrm{rank}}

\newcommand{\cU}{{\mathcal U}}

\newtheorem{lemma}{Lemma}
\newtheorem{conj}{Conjecture}
\newtheorem{definition}{Definition}

\newtheorem{proposition}{Proposition}

\newcommand{\cC}{{\cal C}}

\newcommand{\cP}{{\cal P}}

\def\so{{\mathfrak{so}}}

\newcommand{\la}{\langle}
\newcommand{\JS}{Jordan-Schwinger}
\newcommand{\ra}{\rangle}

\def\one{{\bf 1}}
\def\1{{\bf 1}}
\def\0{{\bf 0}}

\def\dr{{\rightarrow}}

\def\ov{\overline}
\def\conj{\overline}

\def\ii{{i}}
\def\rd{{\textrm{d}}}
\def\demi{{\frac{1}{2}}}
\def\half{{\frac{1}{2}}}
\def\shalf{{\frac{1}{2}}}
\def\Star{{^*}}
\def\Dagger{^\dagger}
\def\Transpose{^{\textrm{t}}}
\def\transpose{{\textrm{t}}}

\newcommand{\act}{{\,\triangleright\,}}
\newcommand{\ract}{{\,\triangleright\,}}

\def\setst{{\,|\,}}

\DeclareMathOperator{\tr}{tr}

\DeclareMathOperator{\SO}{SO}
\DeclareMathOperator{\SL}{SL}
\DeclareMathOperator{\GL}{GL}
\DeclareMathOperator{\U}{U}

\def\poi#1{\{ #1 \}}
\def\set#1{\{ #1 \}}

\def\paren#1{\left( #1 \right)}
\def\ket#1{| #1 \ra}
\def\ketb#1{| #1 ]}
\def\ketp#1{| #1 )}

\def\bra#1{\la #1 |}

\def\braket#1{\la #1 \ra}
\def\praket#1{( #1 )}
\def\braketlb#1{[ #1 \ra}
\def\braketrb#1{\la #1 ]}
\def\bracks#1{\left[ #1 \right]}
\def\abs#1{| #1 |}

\begin{document}
\title{\Large \bf notes}

\title{ $\SO^*(2N)$ coherent states for loop quantum gravity}

\author{{\bf Florian Girelli}}\email{fgirelli@uwaterloo.ca}
\affiliation{Department of Applied Mathematics, University of Waterloo, Waterloo, Ontario, Canada}
\author{{\bf Giuseppe Sellaroli}}\email{gsellaroli@uwaterloo.ca}
\affiliation{Department of Applied Mathematics, University of Waterloo, Waterloo, Ontario, Canada}
\date{\small\today}

\begin{abstract}
\noindent 
A $\SU(2)$ intertwiner  with $N$ legs can be interpreted as the quantum state of a convex polyhedron  with 
$N$ faces (when working in 3d). We show that the intertwiner Hilbert space  carries a representation of the non-compact group $\SO^*(2N)$. This group can be viewed as the subgroup of the symplectic group $\Sp(4N,\R)$ which preserves the $\SU(2)$ invariance. We construct the associated  Perelomov coherent states and  discuss the notion of semi-classical limit, which is more subtle that we could expect. Our work completes the work by Freidel and Livine \cite{freidel_fine_2010, FreidelLivine2011} which focused on the $\U(N)$ subgroup of $\SO^*(2N)$. 

\end{abstract}

\medskip


\maketitle
  
\tableofcontents  
  
 \section*{Introduction} 
 The spinorial formalism for loop quantum gravity (LQG) \cite{Girelli:2005ii} provides a different way to parameterize the LQG Hilbert space and as such provides promising avenues to address some problems encountered in  the field, such as\footnote{For more references see \cite{Livine:2011gp}.}: how to construct the intertwiner observables when dealing with a quantum group (to introduce  a non-zero cosmological constant) \cite{Dupuis:2013lka}, how to implement the simplicity constraints in a natural way  \cite{Dupuis:2011fz}, how to calculate various types of entropies \cite{freidel_fine_2010,Bianchi:2016tmw}.

\medskip

One of the key results of this formalism is that it provides a closed algebra, spanned by $E_{ab}, F_{ab}, \tilde F_{ab}$, to express any intertwiner observables\footnote{The operators $E_{ab}, F_{ab}, \tilde F_{ab}$ are invariant under the global $\SU(2)$ transformations but are not self-adjoint operators. So strictly speaking there are not observables. However we can construct polynomial functions of these operators which will be self-adjoint.}. This algebra, in fact a Lie algebra, contains $ \u(N)$ as a subalgebra (for a $N$ legged intertwiner) which is generated by the   $E_{ab}$. As a consequence, Freidel and Livine have shown  that the space of $N$-legged intertwiners with fixed total area carries a specific representation of $\U(N)$ \cite{freidel_fine_2010}.  They showed furthermore that a $\U(N)$ coherent states (\`a la Perelomov) could be interpreted as a semi-classical polyhedron with $N$ faces and fixed area \cite{FreidelLivine2011}.     The rest of the algebra has not been  fully studied yet and this is what we intend to do here. 

\medskip

Provided we redefine the observables $E_{ab}$ with respect to the usual convention, the full algebra of observables  is isomorphic as a complex algebra to $\so(2N,\C)$. We look then for the real algebra which would have $\u(N)$ as its compact sub-algebra and such that the  $F_{ab}, \tilde F_{ab}$ are antisymmetric under the permutation $a\leftrightarrow b$.  There is an unique choice \cite{boothby_symmetric_1972}, given by $\so^*(2N)$ which spans a non-compact group $\SO^*(2N)$. This group has not been studied much, in particular its representation theory is not completely known (for $N>1$). However, applications of \( \SO^*(2N) \) in physics have been already considered in the past. For example, it has been suggested to use \( \SO^*(2N) \) as a generalized space-time symmetry or as a dynamical algebra containing \( \SO(3,1) \) \cite{so*_applications}. In our work, we  show that the intertwiner Hilbert space provides an infinite-dimensional representation of the $\so^*(2N)$ Lie algebra, parametrized in terms of the total area. Indeed, if the  $ \u(N)$ observables can be understood as transformations between intertwiners with fixed areas, the left-over of the algebra, spanned by $\tilde F_{ab},  F_{ab}$,  can be interpreted as maps between intertwiners which create or annihilate quanta of area. As such given a   $N$ legged intertwiner with a given total area, any other $N$ legged intertwiner can be obtained from it by a suitable $\SO^*(2N)$ transformation. Said otherwise,  if we think of a $N$-legged intertwiner as parametrized in terms of the states of $2N$ harmonic oscillators, invariant under a global $\SU(2)$ transformation, then any other  $N$ legged intertwiner can be obtained by a symplectomorphism (or Bogoliubov transformations) which preserves the $\SU(2)$ invariance. Hence, we will show that $\SO^*(2N)$ can be seen as the subgroup of $\Sp(4N,\R)$ preserving the $\SU(2)$ invariance.

\medskip

Once we have identified the Lie algebra/group, we can construct a new intertwiner coherent state (for a thorough review on intertwiner coherent states see \cite{Livine:2013tsa}). Note that there are different options to generalize the standard concept of coherent state for an harmonic oscillator. Indeed  the harmonic oscillator coherent state satisfies two key properties: the creation operator acts diagonally on the coherent state and the Heisenberg group acts coherently on the state. It is typically only when dealing with  Heisenberg group like structures that we can have both of these properties at once. To generalize the notion of coherent state to the $\SO^*(2N)$ case, we therefore have the choice: we retain any of these properties to construct the state. The construction of coherent states which diagonalize the creation operators $\tilde F_{ab}$ has been performed in \cite{Dupuis:2011fz}. These states are actually tailored to solve the so-called holomorphic simplicity constraints. The other option, to keep a coherent action of the group, falls into the Gilmore-Perelomov program to construct coherent states \cite{perelomov_article, perelomov_book}.   The group $\SO^*(2N)$ being non-compact makes things a bit easier and in fact these coherent states were very succinctly studied in Perelomov's book \cite{perelomov_book}, albeit not for the intertwiner representation. We provide here the full details of their construction in a different representation than \cite{perelomov_book}. We determine the matrix elements of the generators $E_{ab}, F_{ab}, \tilde F_{ab}$ and their expectations values with respect to these states. 

The construction of a coherent state allows for the study of the semi-classical limit. We expect to recover a convex polyhedron with $N$ faces \cite{Bianchi:2010gc}. We show that this can be the case, with some extra subtleties depending on the matrix $\zeta$ parametrizing the coherent state. This  $N\times N$ matrix being antisymmetric, has a rank which is even, $\rank(\zeta)=2k$, and clearly bounded by $N$.  We will note $\lambda_\alpha^2$, $\alpha=1,..,k$ the eigenvalues of $\zeta^*\zeta$.  If all these eigenvalues are distinct, we obtain a (discrete) family of $k$ polyhedra with $N$ faces. In particular, if $\rank(\zeta)=2$, we recover one polyhedron with $N$ faces as we could expect.  $\lambda_\alpha$ (or more exactly a function of it) defines the total area of each  of the polyhedron $\alpha$.   However if some $\lambda_\alpha$ are identical, we actually get some \textit{continuous} families of polyhedra, each of the polyhedron  having a total area specified by $\lambda_\alpha$.

\medskip

It is interesting that the coherent states we have constructed already appeared  in the literature \cite{Bonzom:2012bn,Freidel:2012ji,Freidel:2013fia,Bonzom:2015ova, Bianchi:2016hmk} due to their nice features to perform calculations. Note however that they were always  defined in terms of a matrix $\zeta$ of rank 2, so that there is no issue with the semi-classical limit. Finally, many of the results presented here, especially regarding the construction of the coherent state, were also presented as part of the PhD thesis \cite{thesis}.

In Section I, we review the different parametrization of a classical convex polyhedron, introducing the classical spinorial formalism. In Section II, we introduce the quantum version of the spinorial formalism, ie the harmonic oscillators representation. We review the construction of the $\U(N)$ coherent states \textit{\`a la Perelomov} unlike what Freidel and Livine did in \cite{FreidelLivine2011}. We discuss in particular the semi-classical limit to identify the classical spinors which parametrize the semi-classical polyhedron. 
We will use the same approach to deal with the  $\SO^*(2N)$ coherent states which we define in Section III. We determine the expectation values of the basic observables and the variance of the (total) area with respect to these states. We also explain how these coherent states can be viewed as a specific class of squeezed states. Finally, we discuss how in the semi-classical limit, we can recover a discrete family of polyhedra and/or a continuous one, depending on the nature of coherent state.

\section{Polyhedron parametrization}\label{sec:spinor}\label{sec:polyhedra}
A polyhedron with $N$ faces in $\R^3$ can be reconstructed from  the $N$ normals  $\vec V_a\in\R^3$ of its faces \cite{minkowski} which satisfy what is called the \textit{closure condition}, 
\be\label{closure}
\cC=\sum_{a=1}^N \vec V_a=\vec 0.
\ee
Kapovich and Milson \cite{KM} introduced a phase space structure on the space of polyhedra for \textit{fixed} areas given by $|\vec V_a|=V_a$. The closure condition \eqref{closure} can then be seen as a momentum map implementing global rotations. Their phase space is  given by the symplectic reduction
\be
\cP_N^{KM}= (S^2\times ..\times S^2)/\!/\SO(3), \quad \vec V_a =V_a \hat v_a
\ee
with Poisson bracket on $S^2$
\be
\poi{V^i_a,V^j_b}=\delta_{ab}{\epsilon^{ij}}_k V^k_a, \quad \poi{V_a,V^j_b}=0, \quad \forall a,b.
\ee
$\cP_N^{KM}$ is a space with dimension $2N-6$. From the loop quantum gravity perspective, it is important to also have  the area as a variable. One of the strengths of the so-called \textit{spinor approach} is to provide such parametrization.   To have a phase space structure, one \textit{usually} extends the Kapovich-Milson phase space by replacing $S^2$ by $\C^2\sim \R^4\ni (\vec V,\phi )$. One of the extra degrees of freedom is the area (ie the norm of the vector) whereas the other\footnote{We  consider a space of even dimension as otherwise we cannot have a proper phase space. } one can be seen as a phase $\phi$. If we note the pair of complex numbers\footnote{We change notation with respect to the usual notation in order to avoid too many indices later.} which we call the spinors\footnote{We have  that $\la z|=(\ov x, \ov y)$ and we will also use $|z]= \left(\begin{array}{c}\ov y\\-\ov x\end{array}\right)$ as well as $[z|=(y,-x)$.}, $|z\ra=\left(\begin{array}{c}x\\y\end{array}\right)\in\C^2$, then the maps between the spinors and the vector/phase variables are the following:
\bes
 \vec V&=& \demi \la z|\vec \sigma|z\ra, \quad \vec \sigma \textrm{ being the Pauli matrices and } \quad |\vec V|=V,\nn\\
|z\ra &= & \frac{e^{i\theta}}{\sqrt{2}}\left(\begin{array}{c}
\sqrt{V+V_z}\\
e^{i\phi}\sqrt{V-V_z}
\end{array}\right), \quad e^{i\phi}=\frac{V_z+iV_y}{\sqrt{V^2-V_z^2}}.
\ees
Hence we see that given $\vec V$ we can reconstruct the spinor up to a phase $\theta$. In the spinorial approach, the polyhedron phase space \cite{Livine:2013tsa} is now given by 
\be
\cP^{spin}_N=\C^{2N}/\!/\SU(2), \qquad \textrm{ with } \poi{z_a,\ov z_b}=-i\delta_{ab}, \textrm{ the other brackets being 0}.
\ee 
The symplectic reduction by $\SU(2)$ is given by the closure constraint momentum map expressed in the spinor variables 
\be\label{closure spinor}
\sum_a^N|z_a\ra\la z_a|= \half\sum_a^N\la z_a| z_a\ra \one.
\ee
One of the key advantages of the spinor formalism is that it allows to construct a closed algebra of observables \cite{Girelli:2005ii}. We introduce the $\SU(2)$ invariant quantities 
\begin{itemize}
\item $e_{ab}=\la z_a|z_b\ra $ which changes the area of the  faces $a$ and $b$ while keeping the total area fixed.  If $a=b$ it provides the value of the area of the face $a$.
\item $\tilde f_{ab}=[z_a|z_b\ra$ which  changes the area of the  faces $a$ and $b$ while adding one unit to  the total area.
\item $ f_{ab}=\la z_a|z_b]$ which  changes the area of the  faces $a$ and $b$ while subtracting one unit to  the total area.
\end{itemize}
Any observable built in terms of the normals  $\vec V_a$ such as the norm $|\vec V_a|$ or the relative angle  $\vec V_a\cdot \vec V_b$ can be defined in terms of these observables. 
\be
|\vec V_a|^2= \frac{1}{4}e_{aa}^2, \quad \vec V_a\cdot \vec V_b= \demi e_{ab}e_{ba}-\frac14 e_{aa}e_{bb}.
\ee
Hence the spinor variables provide a finer parametrization of the polyhedron phase space, a parametrization which furthermore closes in terms of the Poisson bracket, unlike the observables expressed in terms of the normals such as $\vec V_a\cdot \vec V_b$.
\bes
\poi{e_{ab},e_{cd}} &= -i \left(\delta_{cb}e_{ad}-\delta_{ad}e_{cb}\right),
\quad 
\poi{e_{ab},f_{cd}} &= -i \left(\delta_{ad}f_{bc}-\delta_{ac}f_{bd} \right), \quad \poi{f_{ab},f_{cd}}=\poi{\tilde{f}_{\!ab},\tilde{f}_{\!cd}}=0 \nn
\\
\poi{e_{ab},\tilde{f}_{\!cd}} &= -i \left(\delta_{bc}\tilde{f}_{\!ad}-\delta_{bd}\tilde{f}_{\!ac}\right),
\quad 
\poi{f_{ab},\tilde{f}_{\!cd}} &= -i \left(\delta_{db}e_{ca}+\delta_{ca}e_{db}-\delta_{cb}e_{da}-\delta_{da}e_{cb} \right). \label{classical so*}
\ees
 The observables $e_{ab}$ form the classical version of the $\u(N)$ algebra, whereas the $e_{ab}$ together with the $f_{ab}$ and the $\widetilde f_{ab}$ form a $\so^*(2N)$ algebra. We will discuss in more  details these structures in  Section \ref{sec:CS}.

\section{Coherent states for the polyhedron with fixed area: a review}
\subsection{Harmonic oscillators and   intertwiner}
We consider  $2N$ quantum harmonic oscillators $(A_a,B_a)$, with the only non-zero commutators   
 \begin{equation} \label{HO comm}
[A_a,A_b\Dagger]=[B_a,B_b\Dagger]=\one \delta_{ab}, 
\end{equation}
which act on the Fock basis
\begin{equation}
\ket{n_A,n_B}_\text{HO}\equiv \ket{n_A}_\text{HO} \otimes \ket{n_B}_\text{HO},\quad n_A,n_b \in \N.
\end{equation}
These harmonic oscillators are the quantum version of the spinors of Section \ref{sec:polyhedra}. The observable generators are then obtained  by quantizing directly their classical definition. We choose the symmetric ordering so that $\ov z z \dr A^\dagger A + \half$ which leads to the following quantum observables\footnote{This ordering was also noticed in the first footnotes of \cite{FreidelLivine2011}. }.  

\begin{equation}\label{eq:ho_E }
 E_{ab} = A\Dagger_a A_b + B\Dagger_a B_b +\delta_{ab}\one, \quad  F_{ab} = B_a A_b - A_a B_b,\quad 
\widetilde F_{ab} = B\Dagger_a A\Dagger_b - A\Dagger_a B\Dagger_b.
\end{equation}
We emphasize the presence of the $\delta_{ab}\one$ term in the definition of $E_{ab}$ which is not usually present in the spinorial formalism where a different ordering is used. Using the harmonic oscillator commutation relations \eqref{HO comm} allows to recover 
\begin{subequations}
\begin{align}
[E_{ab},E_{cd}] &= \delta_{cb}E_{ad} - \delta_{ad}E_{cb}, \quad
[E_{ab},\widetilde F_{cd}]	= \delta_{bc}\widetilde F_{ad} - \delta_{bd}\widetilde F_{ac}, \quad
[E_{ab},F_{cd}]	= \delta_{ad}F_{bc} - \delta_{ac}F_{bd}, \\
[F_{ab},\widetilde F_{cd}]&= \delta_{db}E_{ca} + \delta_{ca}E_{db}  - \delta_{cb}E_{da} -\delta_{da}E_{cb}, \quad  
[F_{ab},F_{cd}]		= [\widetilde F_{ab},\widetilde F_{cd}] = 0,
\end{align}
\end{subequations}
It is essential to use this quantization scheme in order to recover this Lie algebra structure which we will identify to be the  $\so^*(2N)$ Lie algebra. 
\smallskip 

It will prove useful to also introduce the notation\footnote{We use the complex conjugate of \( \zeta \) in \( F_\zeta \) to ensure that  \( (F_\zeta)\Dagger = \widetilde F_\zeta \), which will happen when the ($\SO^*(2N)$) representation is unitary as we shall see later.}
\begin{equation}
E_\alpha:=\alpha^{ab}E_{ab},\quad \widetilde F_\zeta := \zeta^{ab}\widetilde F_{ab}, \quad F_\zeta := \ov \zeta^{ab} F_{ab},\quad \alpha,\zeta \in M_N(\C),
\end{equation}
These elements satisfy the commutation relations
\begin{equation}
[E_\alpha,E_\beta] = E_{[\alpha,\beta]},
\quad
[E_\alpha,\widetilde F_\zeta] = \widetilde F_{\alpha \zeta + \zeta \alpha\Transpose}, \quad 
[E_\alpha, F_\zeta] = - F_{\alpha\Star \zeta + \zeta \conj\alpha}, \quad 
[F_w,\widetilde F_\zeta] = E_{(\zeta-\zeta\Transpose)(w-w\Transpose)\Star}.
\end{equation}

\medskip

The action of observable generators on the intertwiner follows from the Schwinger-Jordan representation of $\su(2)$ representations. Explicitly, we realize an intertwiner  in terms of the harmonic oscillator representations, which allows in turns to have an action of  observable generators on the intertwiner space.

\smallskip

The $\su(2)$ generators are realized in terms of harmonic oscillators as
\begin{equation}
J_z = \demi(A\Dagger A - B\Dagger B),\quad J_+ = A\Dagger B,\quad J_- = B\Dagger A,
\end{equation}
while the $\su(2)$ irreps are 
\begin{equation}
\ket{j,m}=\ket{j+m,j-m}_\text{HO}= \ket{n_A,n_B}_{HO},\quad m\in\{-j,..,j\}.
\end{equation}
One can easily check that the Casimir can be expressed in terms of the $E$ operator.
\begin{equation}
J^2 = \tfrac14 (E-\one)(E+\one),\quad E:=A\Dagger A + B\Dagger B + \one,
\end{equation}
with
\begin{equation}
E\ket{j,m}=(2j+1)\ket{j,m},
\end{equation}
that is, in some sense, \( E \) provides (almost) a square root of the Casimir.
We  extend this construction to the intertwiner space as follows. We denote by \( \operatorname{Inv}_{\SU(2)}(H_{j_1}\otimes\dotsb \otimes H_{j_N}) \) the set of \( \SU(2) \) invariant vectors in the tensor product of \( N \) \( \SU(2) \) irreducible unitary representations, that is those that are annihilated by the \textit{total angular momentum}
\begin{equation}
\vec{J}:=\sum_{a=1}^N \vec{J}^{(a)},
\end{equation}
which we can identify with \( N \)-legged intertwiners. We then introduce the \JS\ representation for each leg, i.e., we use \( 2N \) harmonic oscillators\footnote{It is implicitly assumed that the operators with subscript \( a \) only act on \( H_{j_a} \).}
\begin{equation}
J^{(a)}_z=\half \left(A\Dagger_a A_a - B\Dagger_a B_a\right), \quad J^{(a)}_+ = A\Dagger_a B_a,\quad J^{(a)}_- = B\Dagger_a B_a.
\end{equation}
These vector operators can be seen as the quantization of the polyhedron normals $\vec V_a$.

\smallskip

The $E_{ab}$ satisfy the commutation relations
\begin{equation}
[E_{ab},E_{cd}] = \delta_{cb} E_{ad} - \delta_{ad} E_{cb},
\end{equation}
which are those of a \( \mathfrak{u}(n)_\C \) algebra. These operators can be used to construct all the usual LQG observables, namely
\begin{equation}
\vec{J}^{(a)}\cdot \vec{J}^{(b)} \equiv 2 \mathcal{A}_{ab}\mathcal{A}_{ba} - \mathcal{A}_a \mathcal{A}_b - (1-2\delta_{ab})\mathcal{A}_a,
\end{equation}
where
\begin{equation}
\mathcal{A}_{ab}:= \half (E_{ab}-\delta_{ab}\one),\quad \mathcal{A}_a:=\mathcal{A}_{aa}.
\end{equation}
 We are going to interpret the eigenvalues of the operator \( \mathcal{A}_a \)
\begin{equation}
\mathcal{A}_a\ket{j_a,m_a}= j_a \ket{j_a,m_a}
\end{equation}
as the \textit{area} associated to the leg \( a \), hence we will refer to  the \( \mathcal{A}_a \)'s as \textit{area operators}. The operator \( \mathcal{A}:= \sum_a \mathcal{A}_a \) gives  the total area of the intertwiner.

\subsection{Intertwiner as $\U(N)$ representation}

It was shown in \cite{freidel_fine_2010} that the space of intertwiners with a fixed total area\footnote{The fact that the total area must be an integer follows from the selection rules of the addition of angular momenta.} \( J \in \N\)
\begin{equation}
\mathcal{H}^J_N = \bigoplus_{\sum_{a}j_a=J} \operatorname{Inv}_{\SU(2)}(V_{j_1}\otimes\dotsb \otimes V_{j_N})
\end{equation}
has the structure of an irreducible unitary representation of \( \U(N) \), whose infinitesimal action is given by the \( E_{ab} \) operators we defined\footnote{We recall that our definition differs from that of \cite{freidel_fine_2010}, namely our \( E_{ab} \) have an additional \( \delta_{ab} \) term, which as we will see is essential to construct the \( \SO^*(2N) \) representation.}. Explicitly,
\begin{equation}
\mathcal{H}^J_N \equiv [J+1,J+1,1,\dotsc,1],
\end{equation}
where the \( [\lambda_1,\lambda_2,\dotsc,\lambda_N] \), with
\begin{equation}
\lambda_1\geq \lambda_2 \geq\dotsb \geq \lambda_N\geq 0,
\end{equation}
denotes the \( U(N) \) representation with highest weight vector \( \ket{\lambda} \), for which
\begin{equation}
E_{aa}\ket{\lambda}=\lambda_a\ket{\lambda} \quad\mbox{and}\quad E_{ab}\ket{\lambda}=0,\quad \forall a<b.
\end{equation}
This particular choice of \( \lambda \)'s is required  for the \( \SU(2) \) invariance. The dimension of \( \U(N) \) representations can be computed with the \textit{hook-length formula} \cite{Iachello_2015}
\begin{equation}
\dim [\lambda_1,\dotsc,\lambda_N]= \prod_{a<b}\frac{\lambda_a - \lambda_b + b-a}{b-a},
\end{equation}
which in our specific case gives
\begin{equation}
\dim [\lambda_1,\lambda_2,1,\dotsc,1]= \frac{\lambda_1 - \lambda_2 +1}{\lambda_1} \binom{\lambda_1+N-2}{\lambda_1-1} \binom{\lambda_2+N-3}{\lambda_2-1},
\end{equation}
so that
\begin{equation}
\dim \mathcal{H}^J_N = \frac{1}{J+1} \binom{J+N-1}{J}\binom{J+N-2}{J}=\frac{(N+J-1)!(N+J-2)!}{J!(J+1)!(N-1)!(N-2)!},
\end{equation}
which is indeed the dimension of the space of \( N \)-legged intertwiners with \textit{fixed total area}.

\subsection{$\U(N)$ coherent states}
We will now revisit the construction of \( \mathrm{U}(N) \) coherent states for the intertwiner representation, originally presented in \cite{FreidelLivine2011}. 

\subsubsection{$\U(N)$ coherent states \`a la Perelomov}

Working with the representation \( \mathcal{H}^J_N \), we will use the highest weight vector (the N legged intertwiner where only 2 legs have a non-zero area)
\begin{equation}\label{eq:highest_weight}
\ket{\psi_J}:=\frac{1}{\sqrt{J!(J+1)!}}\widetilde{F}_{12}^J\ket{0}
\end{equation}
as our fixed state. One can easily check that this is indeed the highest weight, i.e.,
\begin{equation}
E_{ab}\ket{\psi_J}=0,\quad \forall a<b. 
\end{equation}
The isotropy subgroup of \(\ket{\psi_J} \) is given by \( \mathrm{U}(2)\times \mathrm{U}(N-2) \), so that, following Perelomov, the coherent states are going to labelled by elements of the quotient space
\begin{equation}
\frac{\mathrm{U}(N)}{\mathrm{U}(2)\times \mathrm{U}(N-2)},
\end{equation}
which is isomorphic to the \textit{Grassmannian}
\begin{equation}
\mathrm{Gr}_2(\C^n) = \set{ \xi\in\so(N,\C) \setst \rank(\xi)=2 }/\sim
\textrm{ where } 
\xi \sim \chi \,\Leftrightarrow\, \xi = \lambda \chi,\quad 0 \neq \lambda\in\C.
\end{equation}
\( \mathrm{U}(N) \) acts on the equivalence classes \( [\xi]\in\mathrm{Gr}_2(\C^N) \) as
\begin{equation}
g \ract [\xi] = \bracks{g \xi g^\transpose}.
\end{equation}
The equivalence class with representative
\begin{equation}
\xi_0 = \begin{pmatrix}
\sigma & 0 \\ 0 & 0
\end{pmatrix}, \quad \textrm{with } \sigma=\begin{pmatrix}
0 & -1 \\ 1 & 0
\end{pmatrix}
\end{equation}
satisfies
\begin{equation}
g\ract [\xi_0] = [\xi_0]
\quad\Leftrightarrow\quad
g\in \mathrm{U}(2) \times \mathrm{U}(N-2).
\end{equation}
For every \( \xi \) there is a (non-unique) unitary matrix, which we will denote by \( g_{[\xi]} \), such that
\begin{equation}
\xi = \lambda \,g_{[\xi]}
\begin{pmatrix}
\sigma & 0 \\ 0 & 0
\end{pmatrix}
g_{[\xi]}^\transpose
\end{equation}
for some \( \lambda \). The notation \( g_{[\xi]} \) is consistent since for each \( \chi \in [\xi]\) we can use the same unitary matrix in the factorisation. We then have
\begin{equation}
[\xi] = \bracks{g_{[\xi]}\xi_0 g_{[\xi]}^\transpose} = g_{[\xi]}\ract [\xi_0].
\end{equation}
We are now in a position to define the coherent states. For each \( [\xi]\in \mathrm{Gr}_2(\C^N) \) we define the state
\begin{equation}
\ket{J,\xi} = \mathcal{N}_J(\xi) \paren{\half\widetilde{F}_{\xi}}^J \ket{0},\quad \mathcal{N}_J(\xi) = \frac{\paren{\shalf \tr(\xi^*\xi)}^{-\frac{J}{2}}}{\sqrt{J!(J+1)!}},
\end{equation}
which one can check to be normalised to \( 1 \), see the end of Appendix  \ref{proof:probability}. Note that the state does not depend on the representative \( \xi \), as
\begin{equation}
\ket{J,\lambda \xi} = \frac{\lambda^J}{\abs{\lambda}^J} \ket{J,\xi} = e^{\ii \theta(\lambda)} \ket{J,\xi},\quad \forall \lambda\neq 0.
\end{equation}
Moreover, we have
\begin{equation}
\ket{J,\xi_0} \equiv \ket{\psi_J}.
\end{equation}
To show that these states are indeed Perelomov coherent states, we have to show that they arise from the action of the group on the state \( \ket{\psi_J} \). To do so, we are going to show a more general result: instead of showing the coherence under the group $\U(N)$, we are going to show the coherence under  \( \GL(N,\C)\) which contains $\U(N)$ as a subgroup. 

\begin{proposition}\label{prop:action of U(n) on highest weight} \cite{FreidelLivine2011}
The action of \( \GL(N,\C)\cong \mathrm{U}(N)_\C \) on the highest weight vector \( \ket{\psi_J} \) is
\begin{equation*}
g\ket{\psi_J}=\frac{\det(g)}{\sqrt{J!(J+1)!}}\paren{\half \widetilde{F}_{g\xi_0 g^\transpose}}^J \ket{0}.
\end{equation*}
\end{proposition}
For the proof see Appendix \ref{proof-prop1}. It follows in particular that
\begin{equation}\label{bibi}
g_{[\xi]} \ket{\psi_J} = \frac{\det(g_{[\xi]})}{\sqrt{J!(J+1!)}} \paren{\half \widetilde{F}_{g_{[\xi]}\xi_0 g_{[\xi]}^\transpose}}^J |0\ra = 
\frac{\det(g_{[\xi]})}{\sqrt{J!(J+1!)}} \lambda^{-J} \paren{\half \widetilde{F}_{\xi}}^J|0\ra
\end{equation}
where
\begin{equation}
\abs{\lambda}^2 = \half \tr(\zeta^*\zeta),
\end{equation}
that is
\begin{equation}
\ket{J,\xi} = e^{\ii \theta(\xi)} g_{[\xi]}\ket{\psi_J}.
\end{equation}
The coherence under $\U(N)$, up to a phase, follows then naturally. 

\subsubsection{Matrix elements and semi-classical limit}

We will compute the matrix elements of the \( \mathfrak{u}(n)_\C \) generators in the coherent state basis following the procedure used in \cite{FreidelLivine2011, Livine:2013tsa}.
Let \( \ketp{J,\xi} \) be the unnormalised coherent state
\begin{equation}
\ketp{J,\xi} = \frac{1}{\mathcal{N}_J(\xi)}\ket{J,\xi}.
\end{equation}
We know from the proof of \eqref{prop:action of U(n) on highest weight} that, for any \( \alpha \in M_n(\C) \),
\begin{equation}
\praket{J,\eta|e^{E_\alpha}|J,\xi} = e^{\tr(\alpha)}\praket{J,\eta|J,e^\alpha\xi e^{\alpha^\transpose}}
=
J!(J+1)!\,e^{\tr(\alpha)}\bracks{\half \tr\paren{\eta^* e^\alpha \xi e^{\alpha^\transpose}}}^J,
\end{equation}
which we can use to find
\begin{equation}
\praket{J,\eta|E_\alpha|J,\xi} = {\der{}{\theta}}\left\{\praket{J,\eta|e^{\theta E_\alpha}|J,\xi}\right\}_{{\theta=0}} = 
J!(J+1)!{\der{}{\theta}} \left\{e^{\theta \tr(\alpha)} \bracks{\half \tr\paren{\eta^* e^{\theta\alpha} \xi e^{\theta\alpha^\transpose}}}^J \right\}_{{\theta=0}}.
\end{equation}
Computing the derivative we find that
\begin{equation}
\braket{J,\eta|E_\alpha|J,\xi} = \mathcal{N}_J(\eta)\mathcal{N}_J(\xi) \praket{J,\eta|E_\alpha|J,\xi} = \braket{J,\eta|J,\xi}\tr(\alpha) + 2J \braket{J-1,\eta|J-1,\xi} \frac{\tr(\eta^*\alpha \xi)}{\sqrt{\tr(\eta^*\eta)\tr(\xi^*\xi)}}.
\end{equation}
In particular, choosing \( \alpha = \Delta_{ab} \) we get
\begin{equation}
\braket{E_{ab}}=\braket{J,\xi|E_{ab}|J,\xi} = \delta_{ab} + 2J \frac{(\xi^*\xi)_{ab}}{\tr(\xi^*\xi)}.
\end{equation}
This expression can be simplified using the fact that any rank-\( 2 \) complex anti-symmetric matrix \( \xi \) can be written as
\begin{equation}\label{eq:xi=UMUt}
\xi = \lambda U
M
U^\transpose,
\textrm{ with } 
M=\lambda \, \sigma \oplus  0_{N-2},  \quad  \lambda = \sqrt{\half \tr(\xi^*\xi)}, \quad \sigma = 
\begin{pmatrix}
0 & -1 \\ 1 & 0
\end{pmatrix}
\end{equation}
and where \( U\in \mathrm{U}(N) \) is a unitary matrix. We can then introduce the \( 2N \) spinors
\begin{equation}
\ket{z_a} = \sqrt{J}
\begin{pmatrix}
U_{a1} \\ U_{a2}
\end{pmatrix}
\end{equation}
satisfying by construction
\begin{equation}
\sum_a \ket{z_a} \bra{z_a} = \sum_{a} \half \braket{z_a|z_a} \,\1_2,\quad \sum_{a}  \braket{z_a|z_a}= 2 J.
\end{equation}
Hence from their definition, the closure constraint is satisfied. Note that as the matrix $\xi$ has rank 2 and is antisymmetric, the equivalence class \( [\xi] \) can be parametrized in terms of \( N \) spinors in many different ways (all related to each other by \( \GL(2,\C) \) transformations). These spinors will however not necessarily satisfy the closure constraint.  This parametrization of $\xi$ was used a lot in \cite{FreidelLivine2011} for doing calculations; in particular, it was discussed how some $\SL(2\C)$ transformation can be used to get them to close. We emphasize that these spinors have nothing to do with the spinors that we used to define the semi-classical limit. The semi-classical spinors we obtained do satisfy the closure constraint, which is expected since after all we are dealing with an intertwiner or a polyhedron, hence  an object invariant under the global $\SU(2)$ transformations.

\medskip 
 
In terms of the spinors we have
\begin{equation}
\braket{J,\xi|E_{ab}|J,\xi} = \delta_{ab} + \braket{z_a|z_b}.
\end{equation}
In a similar fashion, we can compute
\begin{equation}
\braket{J,\xi | E_\alpha E_\beta |J,\xi} = {\der{}{\theta}\der{}{\varphi}}\left\{\braket{J,\xi | e^{\theta E_\alpha} e^{\varphi E_\beta} |J,\xi}\right\} _{{\theta=0, \,\varphi=0}}
\end{equation}
to find variances and covariances. We will concentrate on the area operators
\begin{equation}
\mathcal{A}_a = \half (E_{aa}-\1),
\end{equation}
for which we find
\begin{equation}
\Cov(\mathcal{A}_a,\mathcal{A}_b) = \braket{\mathcal{A}_a\mathcal{A}_b} - \braket{\mathcal{A}_a}\braket{\mathcal{A}_b} = 
\frac{\delta_{ab}}{4}\braket{z_a|z_a} + \frac{1}{4J}\braketrb{z_b|z_a}\braketlb{z_a|z_b} - \frac{1}{4J} \braket{z_a|z_a}\braket{z_b|z_b}
\end{equation}
and
\begin{equation}
\Var(\mathcal{A}_a) = \Cov(\mathcal{A}_a,\mathcal{A}_a) = 
\frac{1}{4}\braket{z_a|z_a}  - \frac{1}{4J} \braket{z_a|z_a}^2.
\end{equation}
Note that both \( \braket{\mathcal{A}_a} \) and \( \Var(\mathcal{A}_a) \) are of order \( 1 \) in \( J \), so that the coefficient of variation \( \frac{\sqrt{\Var(\mathcal{A}_a)}}{\braket{\mathcal{A}_a}} \) approaches~\( 0 \) when the total area \( J \) is large.
We can thus think of the coherent state \( \ket{J,\zeta} \) as being peaked, in the large \( J \) limit, on the classical geometry obtained by introducing the vectors
\begin{equation}
\vec{V}^a = \half \braket{z_a|{\vec \sigma}|z_a};
\end{equation}
these satisfy
\begin{equation}
\sum_a \vec{V}^a = 0,\quad \abs{\vec{V}^a} = \braket{\mathcal{A}_a},
\end{equation}
so we can think of them as the normal vectors to a polyhedron with \( N \) faces \( f_a \), with \( \operatorname{area}(f_a) = \braket{\mathcal{A}_a} \) and total surface area \( J = \braket{\mathcal{A}}\). Note that our spinors are not unique: the unitary matrix appearing in \eqref{eq:xi=UMUt} is defined up to a transformation
\begin{equation}
U\rightarrow U V,\quad V=
\begin{pmatrix}
X & 0 \\ 0 & Y
\end{pmatrix},
\quad X\in \SU(2), \quad Y \in \mathrm{U}(N-2);
\end{equation}
under the same transformation, the spinors change as
\begin{equation}
\ket{z_a} \rightarrow X^\transpose \ket{z_a},
\end{equation}
while the vectors undergo a global \( \SO(3) \) rotation. These are the natural symmetries of the polyhedron, which is defined only up to a global rotation.

\section{A new coherent state for the $\SU(2)$ intertwiner}\label{sec:CS}
As we have seen in Section \ref{sec:spinor}, the actual algebra of observables given in \eqref{classical so*} is bigger than $\u(N)$. The usual parametrization of the $\u(N)$ generators in the spinorial formalism does not contain the identity. By redefining these generators to include the identity we can identify the commutation relations of $\so(2N,\C)$. We then need to identify the real form of this algebra which contain $\u(N)$. Thankfully, there is only one candidate given by $\so^*(2N)$ \cite{boothby_symmetric_1972}. This Lie algebra and its associated (non-compact) Lie group have not been studied much. For example the full representation theory is not known to the best of our knowledge. As we are going to see in Section \ref{sec:Perelomov}, the intertwiner space provides an infinite-dimensional representation of \( \SO^*(2N) \), thanks to the realization in terms of harmonic oscillators.

After having identified the structure of the algebra of observables we can proceed in constructing the coherent states \`a la Perelomov, study some of their properties and check their semi-classical limit. Note that we can also construct different coherent states, not of the Perelomov type, by requiring not their coherence under the group action, but instead the "creation operators" $\widetilde F_{ab}$ to act diagonally on them \cite{Dupuis:2011fz}. Such states allow to solve the simplicity constraints to build some 4d (Euclidian) holomorphic spin foam model \cite{Dupuis:2011fz}.

\subsection{The Lie group $\SO^*(2N)$ and its Lie algebra $\so^*(2N)$}\label{sec:so*_def}
We summarize some of the features of the Lie group $\SO^*(2N)$ and of its Lie algebra that will be useful to construct the Perelomov coherent states.

\subsubsection{The Lie group $\SO^*(2N)$}

Recall that \(\SU(N,N)\) is the group of complex matrices with determinant \(1\) preserving the indefinite Hermitian form
\begin{equation}
\SU(N,N)=\left\{g\in \SL(2N,\C)  , \quad  g^*
\begin{pmatrix}
 \one_N &0 \\ 0&-\one_N 
\end{pmatrix}
g = 
\begin{pmatrix}
 \one_N &0 \\ 0&-\one_N
\end{pmatrix}
\right\}.
\end{equation}

The non-compact Lie group \(G=\SO^*(2N)\) is a subgroup of \(\SU(N,N)\)  such that 
\begin{equation}
\SO^*(2N) = \left\{ g \in \SU(N,N) , \quad  g^t
\begin{pmatrix}
0  & \one_N \\ \one_N & 0
\end{pmatrix}
g = 
\begin{pmatrix}
0  & \one_N \\ \one_N & 0
\end{pmatrix}
\right\}.
\end{equation}

Elements of \( \SO^*(2N) \) can be parametrised as \( 2\times 2 \) block matrices  \cite{perelomov_book}.
\begin{equation}
g=\begin{pmatrix}
A & B \\ -\ov B & \ov A
\end{pmatrix},\quad A,B\in M_N(\C), \quad \textrm{with }  \det(A)\neq 0 . 
\end{equation}
and
\begin{equation}
\label{eq:so*_conditions}
AA\Star-BB\Star =\one, \quad
A\Star A-B\Transpose\ov B =\one,\quad 
A\Star B =- B\Transpose \ov A,\quad 
BA\Transpose = - A B\Transpose,
\end{equation}
and with inverse
\begin{equation}
g^{-1}=
\begin{pmatrix}
A\Star & B\Transpose\\
-B\Star & A\Transpose
\end{pmatrix}.
\end{equation}
The maximal compact subgroup \(K\subseteq \SO^*(2N)\) is isomorphic to \(\mathrm{U}(N)\), and is given by the elements of the form
\begin{equation}
\begin{pmatrix}
U & 0 \\
0 & \ov{U}
\end{pmatrix},
\quad U\in \mathrm{U}(N).
\end{equation}
The group is non-compact for all \( N\geq 2 \), while \( \SO\Star(2)\cong U(1) \).

\subsubsection{The Lie algebra $\so^*(2N)$}

The Lie algebra of \( \SO\Star(2N) \) is
\begin{equation}
\so\Star(2N) = \left\{V \in \su(N,N) ,\quad  V\Transpose
\begin{pmatrix}
0  & \one_N \\ \one_N & 0
\end{pmatrix}
= -
\begin{pmatrix}
0  & \one_N \\ \one_N& 0
\end{pmatrix}
V
\right\}.
\end{equation}
Its elements are parametrised by \( 2\times 2 \) block matrices
\begin{equation}
V=
\begin{pmatrix}
X & Y \\ - \ov Y & \ov X
\end{pmatrix},\quad X,Y\in M_N(\C), \quad \textrm{with } X\Star=-X,\quad Y\Transpose =- Y.
\end{equation}
Hence \( \dim \so\Star(2N)=N(2N-1) \).
A basis for \( \so\Star(2N)_\C \cong \so(2N,\C) \) is given by the matrices
\begin{equation}
E_{ab}=
\begin{pmatrix}
\Delta_{ab} & 0 \\ 0 & -\Delta_{ba}
\end{pmatrix},\quad
F_{ab}=
\begin{pmatrix}
0 & 0 \\ \Delta_{ab} -\Delta_{ba} & 0
\end{pmatrix},\quad
\widetilde F_{ab}=
\begin{pmatrix}
0 & \Delta_{ab} -\Delta_{ba} \\ 0 & 0
\end{pmatrix},\quad
\end{equation}
where \( a,b=1,\dotsc,n \) and \( \Delta_{ab}\in M_N(\C) \) is the matrix with entries
\begin{equation}
(\Delta_{ab})_{cd}=\delta_{ac}\delta_{bd}.
\end{equation}
The \( E_{ab} \) matrices span the complexification of the subalgebra \( \mathfrak{u}(N) \).
The commutation relations of the \( \so\Star(2N) \) complexified generators are (cf \eqref{classical so*})
\begin{subequations}
\label{eq:so*_comm}
\begin{align}
[E_{ab},E_{cd}] &= \delta_{cb}E_{ad} - \delta_{ad}E_{cb}, \quad
[E_{ab},\widetilde F_{cd}]	= \delta_{bc}\widetilde F_{ad} - \delta_{bd}\widetilde F_{ac}, \quad
[E_{ab},F_{cd}]	= \delta_{ad}F_{bc} - \delta_{ac}F_{bd}, \\
[F_{ab},\widetilde F_{cd}]&= \delta_{db}E_{ca} + \delta_{ca}E_{db}  - \delta_{cb}E_{da} -\delta_{da}E_{cb}, \quad  
[F_{ab},F_{cd}]		= [\widetilde F_{ab},\widetilde F_{cd}] = 0,
\end{align}
\end{subequations}
and unitary representations are those for which
\begin{equation}
E\Dagger_{ab}=E_{ba},\quad F\Dagger_{ab}=\widetilde F_{ab}.
\end{equation}

\subsection{$\SO^*(2N)$ Perelomov coherent states for the intertwiner}\label{sec:Perelomov}
Following Perelomov (see appendix \ref{app:Perelomov} and \cite{perelomov_book}), we have the following definition.
\begin{definition}
The $\SO^*(2N)$ coherent states are parameterized by an (antisymmetric) matrix $\zeta$ such that $\zeta\Star\zeta<\one$. They are given by
\begin{equation}
\ket{\zeta}=\mathcal{N}(\zeta)\exp\paren{\half\widetilde F_\zeta}\ket{0},\quad \mathcal{N}(\zeta)=\det(1-\zeta\Star\zeta)^{\half},
\end{equation}
with the following scalar product  
\begin{equation}
\braket{\omega|\zeta}=\frac{\det(\1-\zeta\Star\zeta)^{\half}\det(\1-\omega\Star\omega)^{\half} }{\det(\1-\omega\Star\zeta)}.
\end{equation}
\end{definition}
In the following subsection, we are going to provide some justifications for this definition. Note that our calculations are different than Perelomov's since we use the harmonic oscillator representation. We will then compute the expectations values of the $\so^*(2N)$ generators in this basis. We will also explain how these coherent states can be understood as a specific class of squeezed vaccua.

\subsubsection{$\SO^*(2N)$ Perelomov coherent states} 
For the particular case of the intertwiner representation of \( \SO^*(2N) \), we will choose the harmonic oscillator vacuum \( \ket 0 \) as our fixed state. It is easy to see that the isotropy subgroup for \( \ket{0} \) is the maximal compact subgroup \( K=\mathrm{U}(N)\subset \SO^*(2N) \); the coset space \( \SO^*(2N)/\mathrm{U}(N) \) can be identified with one of the \textit{bounded symmetric domains} classified by Cartan  (see \cite{thesis} for further details), namely
\begin{equation}
\SO^*(2N)/\mathrm{U}(N)\cong\Omega_N:=\set{\zeta\in M_N(\C)\setst \zeta\Transpose=-\zeta \mbox{ and } \zeta\Star\zeta<\one},
\end{equation}
on which \(  \SO^*(2N)  \) acts holomorphically and transitively as
\begin{equation}\label{eq:so*_action_bounded}
g\act \zeta\equiv
\begin{pmatrix}
A & B \\ C & D
\end{pmatrix}
\act \zeta:=\paren{A\zeta+B}\paren{C\zeta+D}^{-1}.
\end{equation}
The isotropy subgroup\footnote{Here we mean the subgroup of all \( g\in G \) such that \( g(0)=0 \).} at \( \zeta=0 \) is given by \( K \), and the correspondence between \( \Omega_N \) and \( \SO^*(2N)/\mathrm{U}(N) \) is given by
\begin{equation}
\zeta\in\Omega_N \mapsto \set{g\in \SO^*(2N) \setst g\act 0=\zeta}\equiv g_\zeta K  \in \SO^*(2N)/\mathrm{U}(N),
\end{equation}
where\footnote{Here \( \sqrt{M} \) denotes the \textit{unique} positive semi-definite square root of a positive semi-definite matrix \( M \). Recall that, since the square root is unique, we have \((\sqrt{A})\Transpose\equiv \sqrt{A\Transpose}\) and analogous expressions for \(\ov A\), and~\(A\Star\).}
\begin{equation}
g_\zeta:=
\begin{pmatrix}
X_\zeta & \zeta \conj X_\zeta\\
\zeta\Star X_\zeta & \conj X_\zeta
\end{pmatrix},
\quad X_\zeta:=\sqrt{(\one-\zeta\zeta\Star)^{-1}}.
\end{equation}
The new coherent intertwiner states $\ket{\zeta}$ are then given by
\begin{equation}
\ket{\zeta}:=g_\zeta\ket{0}, \quad \zeta\in\Omega_N.
\end{equation}
Note how
\begin{equation}
\ket{\zeta}\equiv \ket{g_\zeta\act 0}
\,\Rightarrow\,
g\ket{\zeta}=e^{i\theta(g,\zeta)}\ket{g\act \zeta},\quad \forall g\in \SO^*(2N),\forall \zeta\in\Omega_N.
\end{equation}
A more explicit expression for these states can be obtained using the following lemma.
\begin{lemma}[Block \(UDL\) decomposition]\label{lem:UDL}
Any element of \(\SO^*(2N)\) can be decomposed as
\begin{equation}
\begin{pmatrix}
A & B\\-\ov B& \ov A
\end{pmatrix}
=
\begin{pmatrix}
\one & B\ov A^{-1}\\0&\one
\end{pmatrix}
\begin{pmatrix}
(A\Star)^{-1} &0\\0&\ov A
\end{pmatrix}
\begin{pmatrix}
\one &0\\ -\ov A^{-1}\ov B &\one
\end{pmatrix}
=\exp\paren{\half \widetilde F_{B\ov A^{-1}}}
\exp\paren{E_L}
\exp\paren{-\half F_{ A^{-1} B}}
\end{equation}
where 
\begin{equation}
\exp(E_L)=
\begin{pmatrix}
e^L & 0\\ 0 & e^{-L\Transpose}
\end{pmatrix}
=
\begin{pmatrix}
(A\Star)^{-1} &0\\0&\ov A
\end{pmatrix},
\end{equation}
\end{lemma}
Note that, unless \( B=0 \), the factors do not belong to \(\SO^*(2N)\) anymore, but to its \textit{complexification} \( \SO(2N,\C) \) instead.

As a consequence of Lemma \ref{lem:UDL} we can rewrite \( g_\zeta \) as
\begin{equation}
g_\zeta=\exp\paren{\half \widetilde{F}_\zeta}\exp(E_L)\exp\paren{-\half F_{X_\zeta^{-1} \zeta \ov X_\zeta}}
\end{equation}
where \( L \) is such that
\begin{equation}
e^L=\sqrt{\1-\zeta\zeta\Star}.
\end{equation}
Since \( \ket{0} \) is annihilated by every \( F_{ab} \) and
\begin{equation}
e^{ E_L}\ket{0}=e^{\tr L}\ket{0}=\det(e^L)\ket{0}=\det(1-\zeta\Star\zeta)^{\half}\ket{0},
\end{equation}
we can eventually write the coherent states as
\begin{equation}
\ket{\zeta}=\mathcal{N}(\zeta)\exp\paren{\half\widetilde F_\zeta}\ket{0},\quad \mathcal{N}(\zeta)=\det(1-\zeta\Star\zeta)^{\half}.
\end{equation}
This parametrization allows to relate the $\SO^*(2N)$ coherent states to the $\U(N)$ coherent states, when $\rank(\zeta)=2$: they are just a linear superposition of $\U(N)$ coherent states.

\medskip

Let us now determine their scalar product. Using the fact that the representation is unitary, we can write the inner product between two coherent states as
\begin{equation}
\braket{\omega|\zeta}=\braket{0|g_\omega^{-1} g_\zeta|0},
\end{equation}
with
\begin{equation}
g_\omega^{-1} g_\zeta=
\begin{pmatrix}
X_\omega(\1-\omega \zeta\Star)X_\zeta & X_\omega (\zeta-\omega)\conj X_\zeta\\
\conj X_\omega (\zeta\Star-\omega\Star) X_\zeta & \conj X_\omega(\1-\omega\Star \zeta)\conj X_\zeta
\end{pmatrix}
\end{equation}
which automatically ensures
\begin{equation}
\det(\1-\omega\Star\zeta)\neq 0,
\end{equation}
as \(\conj X_\omega(\1-\omega\Star \zeta)\conj X_\zeta\) must be invertible. We know from Lemma \ref{lem:UDL} that the group element can be written as
\begin{equation}
g_\omega^{-1} g_\zeta=
\exp\paren{ \widetilde F_\alpha}
\exp\paren{ E_\Lambda}
\exp\paren{ F_\beta}
\end{equation}
for some \(\alpha\) and \(\beta\), with \( \Lambda \) such that
\begin{equation}
e^\Lambda=X_\omega^{-1}(\1-\zeta\omega\Star)^{-1}X_{\omega}^{-1}=\sqrt{\1-\zeta\zeta\Star}(\1-\zeta\omega\Star)^{-1}\sqrt{\1-\omega\omega\Star},
\end{equation}
so that
\begin{equation}
\braket{\omega|\zeta}=\det(e^\Lambda)\braket{0|0}=\frac{\det(\1-\zeta\Star\zeta)^{\half}\det(\1-\omega\Star\omega)^{\half} }{\det(\1-\omega\Star\zeta)};
\end{equation}
the Cauchy--Schwarz inequality ensures that
\begin{equation}
\abs{\braket{\omega|\zeta}}^2\leq 1,
\end{equation}
where the equality only holds when \(\omega=\zeta\), as by definition states labelled by different cosets are not proportional to each other. 

\subsubsection{Expectation values of observables}\label{sec:expectation}
Let us determine some properties of these coherent states by looking at the matrix elements of the observables and some of their implications.
\begin{proposition}\label{prop:coherent_matrix_elements}
The matrix elements of the \( \so\Star(2N) \) generators in the coherent state basis are given by
\bes
\braket{\omega|E_{ab}|\zeta}& =&\braket{\omega|\zeta}\big[ \1 + 2 \omega\Star\zeta (\1-\omega\Star\zeta)^{-1}\big]_{ab},\quad 
\braket{\omega|F_{ab}|\zeta}=\braket{\omega|\zeta}\big[2 \zeta(\1-\omega\Star\zeta)^{-1}\big]_{ab},\nn\\
\braket{\omega|\widetilde F_{ab}|\zeta}
&=&\braket{\omega|\zeta}\big[2 (\1-\omega\Star\zeta)^{-1}\conj \omega\big]_{ab}.\label{matrix elements}
\ees\end{proposition}
The proof of this proposition can be found in the appendix \ref{app:mat elements}. From this proposition, we can determine the expectation value and variance of the area observables.
\begin{proposition}[Expectation values of areas]\label{prop:expected_values}
The expectation values of the area operators in a particular coherent state \( \ket{\zeta} \) are
\begin{equation*}
\braket{\mathcal{A}_a}=\bracks{\zeta\Star\zeta\paren{\1-\zeta\Star\zeta}^{-1}}_{aa},
\quad
\braket{\mathcal{A}}=\tr\bracks{\zeta\Star\zeta\paren{\1-\zeta\Star\zeta}^{-1}} ={\tr\bracks{\sigma-\1}} , \quad \textrm{with }  \sigma := \paren{\1-\zeta\Star\zeta}^{-1} 
\end{equation*}
and their variance is
\begin{equation*}
\Var(\mathcal{A}_a)=\half\braket{\mathcal{A}_a}\paren{\braket{\mathcal{A}_a}+1},
\quad
\Var(\mathcal{A})=\sum_{a,b}\braket{\mathcal{A}_{ab}}\paren{\braket{\mathcal{A}_{ab}}+\delta_{ab}}={\tr(\sigma(\sigma-\1))}.
\end{equation*}
Moreover, when the non-zero eigenvalues of \( \zeta\Star\zeta \) approach \( 1 \), although \(\Var(\mathcal{A})\) grows without bound, the {coefficient of variation} \( \frac{\sqrt{\Var(\mathcal{A})}}{\braket{\mathcal{A}}} \) approaches a value in \( (0,1] \).
\end{proposition}
The proof of this proposition can also be found in the appendix \ref{app:mat elements}.
Let us spend  few words on the last result of Proposition \ref{prop:expected_values}, regarding the coefficient of variation. This coefficient measures the \textit{relative} standard deviation, i.e., the amount of dispersion compared to the value of the mean. In our particular case, the result is telling us that, even though the dispersion gets bigger as the total area increases, the relative standard deviation is bounded by a value that approaches \( 1 \) for sufficiently large area. 
Note that the coefficient of variation does not provide any useful information when the area is very small, as\footnote{Using the fact that, as \( \sigma \geq 0  \), \(\tr(\one) \tr(\sigma^2)\geq \tr(\sigma)^2 \).}
\begin{equation}
\textrm{When }  \braket{\mathcal{A}} \rightarrow 0, \quad 
\frac{\sqrt{\Var{\mathcal{A}}}}{\braket{\mathcal{A}}}=\frac{\sqrt{\tr\bracks{\sigma(\sigma-\1)}}}{\tr(\sigma-\1)}\geq
\frac{\sqrt{\frac{1}{N}\tr(\sigma)\tr(\sigma-\1)}}{\tr(\sigma-\1)}\rightarrow \infty.
\end{equation}

In the specific case when \( \rank(\zeta)=2 \) we can do much more than computing expectation values and variances: in fact, we can produce the complete probability distribution of the total area as follows\footnote{When \( \rank(\zeta)>2 \) an important simplifying assumption is missing, namely,  that \( \zeta\zeta\Star\zeta  \) is proportional to \( \zeta \).}.
\begin{figure}[h]
\includegraphics[scale=.7]{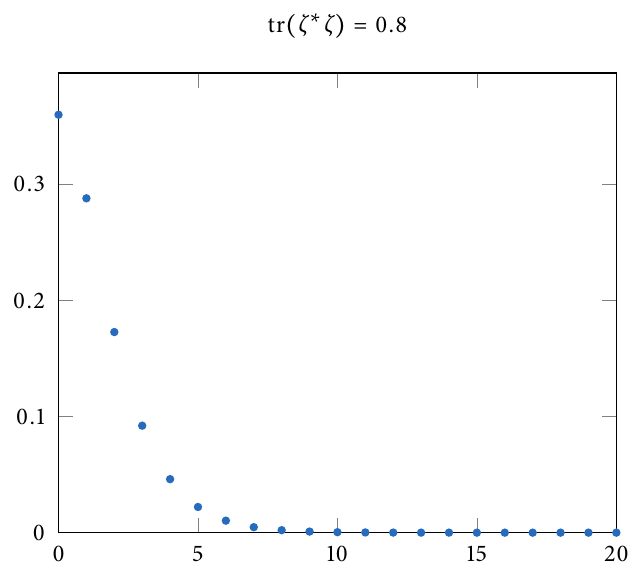}
\includegraphics[scale=.7]{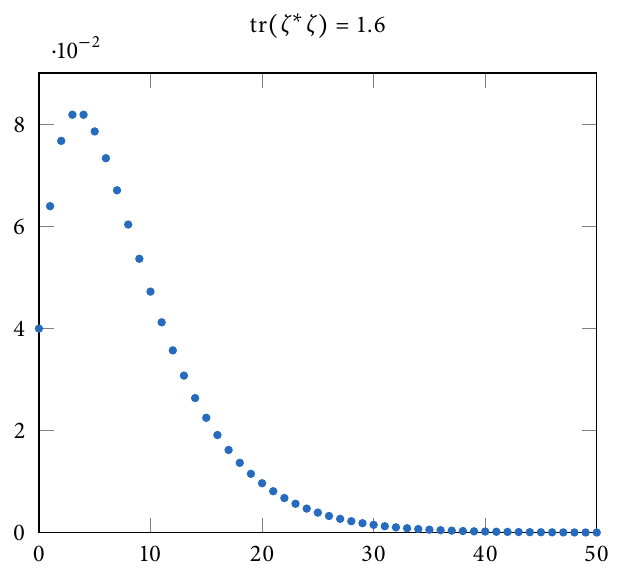}
\includegraphics[scale=.7]{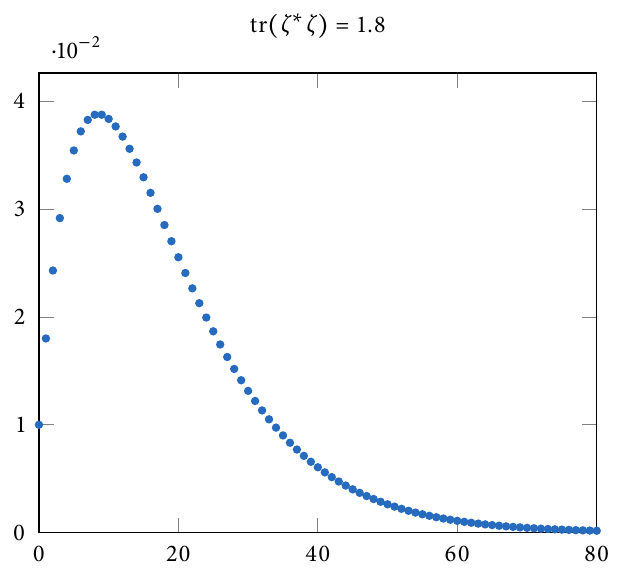}
\includegraphics[scale=.7]{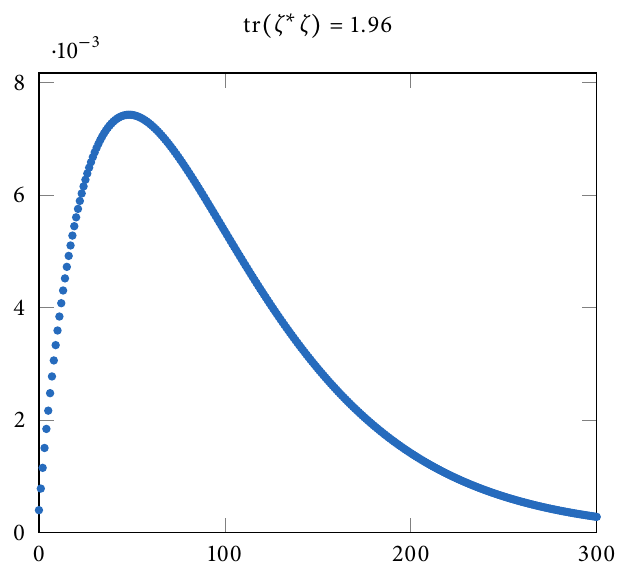}
\caption{Distribution of total area for different values of \( \tr(\zeta\Star\zeta)\) when $\rank(\zeta)=2$. }\label{fig1}
\end{figure}
\begin{proposition}[Probability distribution of total area]\label{prop:probability}
When \(\zeta\) is of rank \(2\) the probability distribution for the total area in the state \(\ket{\zeta}\) is
\[
P_\zeta(J)=\det(\1-\zeta\Star\zeta)\paren{\half\tr(\zeta\Star\zeta)}^J (J+1),\quad J\in \N_0.
\]
\end{proposition}
Proof of this proposition can be found in Appendix \ref{proof:probability}. Plots for the probability distribution can be found in Fig.~\ref{fig1}.
Note how, as  the non-zero eigenvalue of \( \zeta\Star\zeta \) approaches \( 1 \) (or equivalently \( \tr(\zeta\Star\zeta)\rightarrow 2 \)), the relative shape of the distribution remains the same, which is a consequence of Proposition \ref{prop:expected_values}.

\subsubsection{Relating  \( \SO^*(2N) \) and some Bogoliubov transformations}
Some recent works by Bianchi and collaborators have emphasized the use of the symplectic group \( \Sp(4N,\R) \) to  recover some interesting features for loop quantum gravity, such as anew parametrization of the loop variables \cite{Bianchi:2016hmk, Bianchi:2016tmw}. We would like now to  coherent relate our states to this approach.
  
In fact, these  states can be reinterpreted in terms of \textit{Bogoliubov transformations} by making use of the connection between \( \SO^*(2N) \) and the \textit{symplectic group} \( \Sp(4N,\R) \). Recall that, if we have a set of \( n \) decoupled harmonic oscillators
\begin{equation}
[C_a,C\Dagger_b]=\delta_{ab},\quad [C_a,C_b]=[C\Dagger_a,C\Dagger_b]=0,
\end{equation}
a Bogoliubov transformation is a a canonical transformation which maps them to a new set of harmonic oscillators,
\be
\begin{pmatrix}
\widetilde C \\ \widetilde C\Dagger
\end{pmatrix}
=
\begin{pmatrix}
U & V \\ \ov V & \ov U
\end{pmatrix}
\begin{pmatrix}
C \\ C\Dagger
\end{pmatrix} 
\ee 
The conditions on \( U \) and \( V \) such that
\begin{equation}
[\widetilde C_a,\widetilde C\Dagger_b]=\delta_{ab},\quad [\widetilde C_a,\widetilde C_b]=[\widetilde C\Dagger_a,\widetilde C\Dagger_b]=0
\end{equation}
are
\begin{equation}\label{eq:sp_conditions}
UU\Dagger - V V\Dagger = \1,\quad U V\Transpose = V U\Transpose,
\end{equation}
which automatically ensure that \( U \) is invertible and that
\begin{equation}
\begin{pmatrix}
U & V \\ \ov V & \ov U
\end{pmatrix} \in \Sp(2n,\R);
\end{equation}
as such, we can interpret \( \Sp(2n,\R) \) as the group of Bogoliubov transformations of \( n \) harmonic oscillators. The vacuum for the set of new harmonic oscillators is given by
\begin{equation}\label{vac}
\ket{\widetilde 0 } := \mathcal{N} \exp\paren{\half S^{ab} C\Dagger_a C\Dagger_b}\ket{0},
\end{equation}
also known as the \textit{squeezed vacuum}, where \( S \) is the symmetric matrix 
\begin{equation}
S= - U^{-1}V.
\end{equation}
In fact, it is easy to see that
\begin{equation}
\begin{split}
C_d \ket{\widetilde 0} &= \mathcal{N} \sum_{k=0}^\infty \frac{1}{k!}\bracks{C_d,\paren{\half S^{ab} C\Dagger_a C\Dagger_b}^k}\ket{ 0}
= \mathcal{N} \sum_{k=0}^\infty \frac{1}{k!} k \paren{\half S^{ab} C\Dagger_a C\Dagger_b}^{k-1}\paren{\half S^{cd} C\Dagger_c + \half S^{dc}C\Dagger_c} \ket{ 0}
\\
&= \mathcal{N} \sum_{k=0}^\infty \frac{1}{k!} \paren{\half S^{ab} C\Dagger_a C\Dagger_b}^{k} S^{dc} C\Dagger_c \ket{ 0}
= S^{dc} C\Dagger_c \ket{\widetilde 0},
\end{split}
\end{equation}
from which it follows that
\begin{equation}
\widetilde C_a \ket{\widetilde 0 } = 0.
\end{equation}
The fact that \( \ket{\widetilde 0} \) has finite norm can be proven by evaluating \( \braket{\widetilde 0|\widetilde 0} \) as a Gaussian integral, making use of the resolution of the identity in terms of the coherent states for the harmonic oscillators \( C_a \).

In our framework, we use $n=2N$ harmonic oscillators, so we expect to deal with \( \Sp(4N,\R) \).   To connect \( \SO^*(2N) \) to the Bogoliubov transformations, note that \( \SO^*(2N) \) can be embedded into \( \Sp(4N,\R) \) as
\begin{equation}
\varphi:
\begin{pmatrix}
X & Y \\ -\ov Y & \ov X
\end{pmatrix}
\in \SO\Star(2N)
\mapsto
\begin{pmatrix}
X & 0 & 0 & -Y \\
0 & X & Y & 0 \\
0 & -\ov Y & \ov X & 0\\
\ov Y & 0 & 0 &\ov X
\end{pmatrix}
\in \Sp(4N,\R).
\end{equation}
Indeed it is a simple exercise to show that the conditions \eqref{eq:so*_conditions} ensure that \eqref{eq:sp_conditions} hold. Hence we can interpret \( \SO^*(2N) \) as a subgroup of Bogoliubov transformations of the \( 2N \) harmonic oscillators \( A_a \), \( B_b \) that we use to construct the \JS\ representation. In particular, for the Bogoliubov transformation \( \varphi(g_\zeta^{-1}) \), with \( \zeta \in \Omega_N \) we get
\begin{equation}
S=\begin{pmatrix}
0 & -\zeta \\ \zeta & 0
\end{pmatrix},
\end{equation}
so that the associated squeezed vacuum \eqref{vac} is
\begin{equation}
\mathcal{N} \exp\paren{\demi S^{ab} B\Dagger_a A\Dagger_b}\ket{0} = \mathcal{N} \exp\paren{\half \zeta^{ab} \widetilde F_{ab}}\ket{0},
\end{equation}
which is exactly the coherent state \( \ket{\zeta} \).
To summarise, we can regard the new coherent intertwiners we defined  as the squeezed vacua associated to a subgroup of Bogoliubov transformations, isomorphic to \( \SO^*(2N) \). The particular Bogoliubov transformations are exactly those for which the squeezed vacuum is still \( \SU(2) \) invariant (i.e, an intertwiner), so that we can essentially regard \( \SO^*(2N) \) as the group of canonical transformations of \( 2N \) harmonic oscillators preserving \( \SU(2) \) invariance, where the \( \SU(2) \) action is implemented through the \JS\ representation.

\subsection{Semi-classical limit}
Let us now consider the semi-classical limit of our coherent intertwiners. Our goal is to obtain out of the expectation values of the algebra generators a set of variables that, endowed with the appropriate Poisson structure, we can interpret as a classical geometry (similarly to what we discussed in Section \ref{sec:polyhedra}). In particular, we want to be able to construct a set of vectors that sum to zero, and as such can be regarded as the normals to a convex polyhedron \cite{return_spinor}.

\subsubsection{Recovering the spinor variables}

In order to investigate the semi-classical limit, it will prove useful to  rewrite the expected values of the \( \so^*(2N) \) generators \eqref{matrix elements} in a different way. Note the similarity with the bra-ket notation we introduced in section \ref{sec:spinor} when working with classical spinors.
\begin{proposition}\label{prop:semiclassical}
The expectation values of the \(\so^*(2N)\) generators can be written in the form
\begin{equation*}
\braket{\zeta|F_{ab}|\zeta}=\sum_{\alpha=1}^{k}\frac{1}{\lambda_\alpha}\braketlb{z^\alpha_a|z^\alpha_b}
,\quad
\braket{\zeta|\widetilde{F}_{ab}|\zeta}=\sum_{\alpha=1}^{k}\frac{1}{\lambda_\alpha}\braketrb{z^\alpha_a|z^\alpha_b}
,\quad
\braket{\zeta|E_{ab}|\zeta}=\delta_{ab}+\sum_{\alpha=1}^{k} \braket{z^\alpha_a|z^\alpha_b},
\end{equation*}
where \( k=\half\rank(\zeta) \),  \(\lambda_\alpha^2\) is a non-zero eigenvalue of \( \zeta\Star\zeta \) and the spinors are specified in terms of a unitary matrix $U$,
\begin{equation}
\ket{z^\alpha_a}= \paren{\frac{2\lambda^2_\alpha}{1-\lambda_\alpha^2}}^{\half}
\begin{pmatrix}
U_{a,2\alpha-1} \\ U_{a,2\alpha}
\end{pmatrix} \textrm{ with } \zeta = U M U\Transpose \textrm{ and } M=\bigoplus_{\alpha=1}^k \lambda_\alpha 
\begin{pmatrix}
0 & -1 \\ 1 & 0
\end{pmatrix}
\oplus \0_{N-2k},\quad \lambda_\alpha>0.
\end{equation}
As such the spinors satisfy automatically some closure constraint
\begin{equation*}
\sum_{a=1}^{n}\ket{z^\alpha_a}\bra{z^\beta_a}=\delta_{\alpha\beta} \sum_{a=1}^n \half \braket{z^\alpha_a|z^\alpha_a}\1_2.
\end{equation*}
\end{proposition}
\begin{proof}
Since $\zeta$ is an antisymmetric matrix, we know that $\zeta = U M U\Transpose$, where $U$ is unitary and
\begin{equation}
M=\bigoplus_{\alpha=1}^k \lambda_\alpha 
\begin{pmatrix}
0 & -1 \\ 1 & 0
\end{pmatrix}
\oplus \0_{N-2k},\quad \lambda_\alpha>0.
\end{equation}
As a consequence,
\begin{equation}
M\Star M=\bigoplus_{\alpha=1}^k \lambda_\alpha^2
\1_2 \oplus \0_{N-2k}
\textrm{ and } 
(\1-M\Star M)^{-1}=\bigoplus_{\alpha=1}^k (1-\lambda_\alpha^2)^{-1}\,
\1_2 \oplus \1_{N-2k}.
\end{equation}
It follows that
\bes
\braket{\zeta|F_{ab}|\zeta} &=& \left[ 2\zeta(\1-\zeta\Star\zeta)^{-1}\right]_{ab}
= \left[ 2 U M (\1-M\Star M)^{-1}U\Transpose \right]_{ab}
 =\sum_{\alpha=1}^{k}\sum_{c,d=1}^N \frac{2\lambda_\alpha}{1-\lambda_\alpha^2}
U_{ac}\paren{\delta_{c,2\alpha}\delta_{d,2\alpha-1}-\delta_{c,2\alpha-1}\delta_{d,2\alpha}}U_{bd}
\nn\\
&=& \sum_{\alpha=1}^{k} \frac{2\lambda_\alpha}{1-\lambda_\alpha^2} \paren{U_{a,2\alpha}U_{b,2\alpha-1}-U_{a,2\alpha-1}U_{b,2\alpha}}
\ees
and
\bes
\braket{\zeta|E_{ab}|\zeta} -\delta_{ab} &=& \left[ 2\zeta\Star\zeta(\1-\zeta\Star\zeta)^{-1}\right]_{ab}
= \left[ 2 \ov{U} M\Star M (\1-M\Star M)^{-1}U\Transpose \right]_{ab}\\
&=& \sum_{\alpha=1}^{k}\sum_{c,d=1}^N \frac{2\lambda_\alpha^2}{1-\lambda_\alpha^2}
\ov{U}_{ac}\paren{\delta_{c,2\alpha-1}\delta_{d,2\alpha-1}+\delta_{c,2\alpha}\delta_{d,2\alpha}}U_{bd}\\
&=& \sum_{\alpha=1}^{k} \frac{2\lambda_\alpha^2}{1-\lambda_\alpha^2} \paren{\ov{U}_{a,2\alpha-1}U_{b,2\alpha-1} + \ov{U}_{a,2\alpha}U_{b,2\alpha}}.
\ees
Choosing
\begin{equation}
\ket{z^\alpha_a}= \paren{\frac{2\lambda^2_\alpha}{1-\lambda_\alpha^2}}^{\half}
\begin{pmatrix}
U_{a,2\alpha-1} \\ U_{a,2\alpha}
\end{pmatrix}
\quad\Rightarrow\quad
\ketb{z^\alpha_a}= \paren{\frac{2\lambda^2_\alpha}{1-\lambda_\alpha^2}}^{\half}
\begin{pmatrix}
\ov{U}_{a,2\alpha} \\ -\ov{U}_{a,2\alpha-1}
\end{pmatrix}
\end{equation}
we find
\begin{equation}
\braket{\zeta|F_{ab}|\zeta}=\sum_{\alpha=1}^{k}\frac{1}{\lambda_\alpha}\braketlb{z^\alpha_a|z^\beta_b},\quad \braket{\zeta|E_{ab}|\zeta}=\delta_{ab} +\sum_{\alpha=1}^{n}  \braket{z^\alpha_a|z^\alpha_b}.
\end{equation}
Moreover, we have the closure constraints 
\bes
\sum_{a=1}^{N}\ket{z^\alpha_a}\bra{z^\beta_a} &=&
\paren{\frac{2\lambda^2_\alpha}{1-\lambda_\alpha^2}}^{\half}
\paren{\frac{2\lambda^2_\beta}{1-\lambda_\beta^2}}^{\half}
\sum_{a=1}^N
\begin{pmatrix}
U_{a,2\alpha-1}\conj{U}_{a,2\beta-1} & U_{a,2\alpha}\conj{U}_{a,2\beta-1}
\\ 
U_{a,2\alpha-1}\conj{U}_{a,2\beta} & U_{a,2\alpha}\conj{U}_{a,2\beta}
\end{pmatrix} \nn
\\
&=& \delta_{\alpha\beta} \frac{2\lambda^2_\alpha}{1-\lambda_\alpha^2} \1_2
= \delta_{\alpha\beta} \sum_{a=1}^N \half \braket{z^\alpha_a|z^\alpha_a}\1_2 \label{closure semi 1}
\ees
as expected.
\end{proof}
As consequence of this fact, in the limit
$
\lambda_\alpha\rightarrow 1,\quad \alpha=1,\dotsc,k
$
where the expected value of the total area
\begin{equation}
\braket{\mathcal{A}}=\sum_{\alpha=1}^{k}\frac{\lambda_\alpha^2}{1-\lambda_\alpha^2}\rightarrow \infty,
\end{equation}
we have
\begin{equation}
\braket{\zeta|F_{ab}|\zeta} \sim \sum_{\alpha=1}^{k} \braketlb{z^\alpha_a|z^\alpha_b}=\sum_{\alpha=1}^{k} f_{ab}^{\alpha},
\quad
\braket{\zeta|E_{ab}|\zeta} = \delta_{ab} + \sum_{\alpha=1}^{k}\braket{z^\alpha_a|z^\alpha_b}= \delta_{ab} +\sum_{\alpha=1}^{k} e_{ab}^{\alpha}.
\end{equation}
We can interpret the semi-classical limit as a classical geometry by introducing the canonical Poisson structure on \( \C^{2kN}\) which is, using the coordinates \( |z^\alpha_a\ra=\left(\begin{array}{c}x^\alpha_a\\y^\alpha_a\end{array}\right)\),
\begin{equation}\label{eq:poisson_ho}
\poi{x^\alpha_a,\ov{x}^\beta_b}=\poi{y^\alpha_a,\ov{y}^\beta_b}= -i \delta^{\alpha\beta}\delta_{ab}
\end{equation}
with all other brackets vanishing. With this Poisson structure, the functions 
\begin{equation}\label{coarsegrained}
e_{ab}:=\sum_{\alpha=1}^{k}\braket{z^\alpha_a|z^\alpha_b} =\sum_{\alpha=1}^{k} e_{ab}^{\alpha},\quad f_{ab}:=\sum_{\alpha=1}^{k}\braketlb{z^\alpha_a|z^\alpha_b}=\sum_{\alpha=1}^{k} f_{ab}^{\alpha}
\end{equation}
satisfy
\begin{subequations}
\begin{align}
\poi{e_{ab},e_{cd}}& = -i \paren{\delta_{cb}e_{ad}-\delta_{ad}e_{cb}}
,\quad 
\poi{e_{ab},f_{cd}} = -i \paren{\delta_{ad}f_{bc}-\delta_{ac}f_{bd} },\quad 
\poi{f_{ab},f_{cd}}=\poi{\conj{f}_{\!ab},\conj{f}_{\!cd}}=0,
\\
\poi{e_{ab},\conj{f}_{\!cd}} &= -i \paren{\delta_{bc}\conj{f}_{\!ad}-\delta_{bd}\conj{f}_{\!ac} }
,\quad 
\poi{f_{ab},\conj{f}_{\!cd}} = -i \paren{\delta_{db}e_{ca}+\delta_{ca}e_{db}-\delta_{cb}e_{da}-\delta_{da}e_{cb} }
,
\end{align}
\end{subequations}
which are the classical analogue of the \( \so^*(2N) \) commutation relations \eqref{eq:so*_comm}.

\subsubsection{Geometric interpretation of the semi-classical states }
We would like now to determine whether the spinorial formalism we have recovered actually allows the reconstruction of a polyhedron or of anything else.

The rank of the matrix parametrizing the coherent state actually determines the number of spinors we recover.  As we already recalled, the rank of  $\zeta$ is at most $N$, and since $\zeta$ is an antisymmetric matrix its rank has to be even. We noted it $2k$. In the semi-classical limit we have recovered  $2kN$ spinors. 

\smallskip

Let us consider the case $k=1$, that is  $\rank(\zeta)=2$ to start. In fact it is this case that was always considered until now in the literature \cite{Bonzom:2012bn,Freidel:2012ji,Freidel:2013fia,Bonzom:2015ova}.  In this case,  we recover exactly the standard spinorial parametrization of the convex polyhedron since we recover $2N$ spinors. Indeed there is no sum to consider in \eqref{coarsegrained} and  \eqref{closure semi 1} states that the closure constraint is satisfied. Hence we recover in this case,   a convex polyhedron in the semi-classical limit.  

\smallskip

The cases $k>1$ are more subtle. Indeed,   we  recover more spinors  than we started with. We can expect different interpretations of the resulting construction, bearing in mind that we have  $k$ closure constraints in \eqref{closure semi 1}.
 
\begin{itemize}
\item We  recover a polyhedron with $kN$ faces, such that the normal of the faces cancel by bundle of $N$ (think of the cube, where the six normals cancel by bundle of two).  Since we can define some total observables  as in \eqref{coarsegrained},  this $kN$ faces polyhedron could be coarse-grained as a polyhedron with $N$ faces.
\item We  recover a set of k polyhedra with $N$ faces, which could be coarse-grained as a polyhedron with $N$ faces thanks to \eqref{coarsegrained},\textit{ or not}. 
\item We do not really recover anything close to a finite set of polyhedra. 
\end{itemize} 
To assess what we really obtained, we need to check what are the symmetry transformations of the set of spinors we have reconstructed.

\begin{proposition}\label{prop:symmetries}
The unitary matrix \( U \) in the decomposition
\begin{equation*}
\zeta=UMU^\transpose
\end{equation*}
is defined up to a unitary transformation $W$
\begin{equation*}
U\rightarrow UW,\quad W\in \bigtimes_{i=1}^\ell \Sp(2\mu_i) \times \mathrm{U}(N-2k),
\end{equation*}
where \( \ell \) is the number of distinct \( \lambda_\alpha \) in the decomposition of \( \zeta \) and the \( \mu_i \) are the multiplicities of each distinct \( \lambda_\alpha \).  $\Sp(2\mu_i)$ is the \textbf{compact} symplectic group\footnote{The compact symplectic group $\Sp(2\mu_i)$ should not be confused with the real symplectic group $\Sp(2\mu_i,\R)$. $\Sp(2\mu_i)$ is the simply-connected, maximal compact real Lie subgroup of the complex symplectic group $\Sp(2\mu_i,\C)$. It also can be seen as $\Sp(2\mu_i) := \Sp(2\mu_i,\C)\cap \mathrm{U}(2\mu_i).$ }. When all the \( \lambda_\alpha \) are distinct this reduces to
\begin{equation*}
V\in \SU(2)^{\times k} \times \mathrm{U}(N-2k).
\end{equation*}
\end{proposition}
Since the proof is a bit lengthy, we postpone it to  Appendix \ref{proof of prop sym}. From the definition of the spinors, the relevant symmetries we need to consider are given by $\Sp(2\mu_i)$. The left-over, given by $\U(N-2k)$, does not affect the definition of the spinors.  The symmetries we have identified, together with the closure constraints \eqref{closure semi 1}, allow us to interpret the nature of the geometric structures we recover in the semi-classical limit. 
\medskip

First note that when there is only one $\lambda$, that is $\rank(\zeta)=2$, we trivially recover only one copy of $\SU(2)$ which is the global symmetry of the polyhedron. Indeed, it is defined only up to a global rotation. Furthermore the observables $e_{ab}, f_{ab}, \tilde f_{ab}$ are invariant under such transformations. This is consistent with the fact that we recovered a single polyhedron.

\medskip 
Second, when all the $\lambda_\alpha$'s are distinct, we have $k$ global $\SU(2)$ symmetries, together with $k$ closure constraints. This indicates that we have in general not \textit{one} polyhedron but \textit{a set of} $k$ polyhedra with $N$ faces. The observables $e_{ab}^{\alpha}, f_{ab}^{\alpha}, \tilde f_{ab}^{\alpha}$ are invariant under each of these global $\SU(2)$ transformations.  For a given $\alpha$, these observables correspond to the observables for a given polyhedron $\alpha$. The new set of observables is hence generated by $k$ copies of  $\so^*(2N)$. The geometry of each of these polyhedra can be reconstructed from the knowledge of the observables as usual.

\smallskip

One might then wonder whether the definition of the diagonal $\so^*(2N)$ from \eqref{coarsegrained} allows to coarse-grain the set of $k$  polyhedra with $N$ faces to a single new $N$ faces polyhedron. As the following proposition shows, the answer is negative: the diagonal    $\so^*(2N)$ observables do not allow to reconstruct a polyhedron that is closed. 

\begin{proposition}\label{prop:no-go}
Consider $k$ convex $N$ faces polyhedra indexed by $\alpha$ and their associated algebra $\so^*(2N)$ of observables spanned by  $e_{ab}^{\alpha}, f_{ab}^{\alpha}, \tilde f_{ab}^{\alpha}$. The diagonal subalgebra $\so^*(2N)$  spanned by $e_{ab}=\sum_{\alpha=1}^k e_{ab}^{\alpha}, f_{ab}=\sum_{\alpha=1}^kf_{ab}^{\alpha}, \tilde f_{ab}=\sum_{\alpha=1}^k\tilde f_{ab}^{\alpha}$ is not the algebra of observables of a single \textit{closed} polyhedron.
\end{proposition}
\begin{proof}
Since we deal with polyhedra, we have the closure constraints 
\bes
\sum_{a=1}^{N}\ket{z^\alpha_a}\bra{z^\beta_a} &=&
 \delta_{\alpha\beta}  \sum_{a=1}^N \half \braket{z^\alpha_a|z^\alpha_a}\1_2  =  \Lambda_\alpha \1_2 \label{closure semi 2}.
\ees
where $\Lambda_\alpha$ is half of the total area of the polyhedron $\alpha$, 
\be
\sum_{a=1}^N  \braket{z^\alpha_a|z^\alpha_a} = \sum_{a=1}^N  e_{aa}^{\alpha} = 2\Lambda_\alpha.
\ee
Let us now consider the normals $\vec V_a$ of the coarse-grained polyhedron. 
By construction, the relative angles between the normals are given by
\be
\vec V_a\cdot \vec V_b= \demi e_{ab}e_{ba}-\frac14 e_{aa}e_{bb}.
\ee
If the closure constraint $\sum_a \vec V_a=0$ is satisfied then we also have that $\sum_{ab}\vec V_a\cdot \vec V_b=0$. Hence we are supposed to show that 
\be
\sum_{ab}\left(\demi e_{ab}e_{ba}-\frac14 e_{aa}e_{bb}\right)=0.
\ee
Using the definition of the coarse-grained $e_{ab}$ in terms of the $e_{ab}^{\alpha}$ we have
\bes
\sum_{ab}e_{ab}e_{ba}&=& \sum_{ab\alpha\beta}\braket{z^\alpha_a|z^\alpha_b}\braket{z^\beta_b|z^\beta_a}=\sum_{a\alpha\beta} \delta_{\alpha\beta}\Lambda_\alpha\braket{z^\alpha_a|z^\beta_a}=2\sum_{\alpha}\Lambda_\alpha^2,\\
\sum_{ab}e_{aa}e_{bb}&=&\left(\sum_ae_{aa}\right)^2= \left(\sum_{a\alpha}e^{\alpha}_{aa}\right)^2=  \left(2\sum_{\alpha}\Lambda_\alpha\right)^2 = 4 \sum_{\alpha\beta}\Lambda_\alpha\Lambda_\beta.
\ees
We deduce then that 
\be\label{bip}
\sum_{ab}\left(\demi e_{ab}e_{ba}-\frac14 e_{aa}e_{bb}\right)=\sum_{\alpha}\Lambda_\alpha^2 - \sum_{\alpha\beta}\Lambda_\alpha\Lambda_\beta = -\sum_{\alpha\neq\beta} \Lambda_\alpha\Lambda_\beta <0,
\ee
since $\Lambda_\alpha>0$. Hence the closure constraint is not satisfied and we do not have a coarse-grained polyhedron. 
\end{proof}
Hence when all the $\lambda_\alpha$'s are distinct, we have a collection of $k$ polyhedra with $N$ faces, which cannot be seen as a unique polyhedron. One could try use a boost to make the spinors close as explained extensively in \cite{Livine:2013tsa}. This however would not close the polyhedron: an overall \( X\in\SL(2,\C) \) transformation of the spinors would preserve the matrix \( \zeta \), but it would only have the effect of transforming the areas as
\begin{equation}
\Lambda_\alpha \rightarrow \Lambda'_\alpha = \sum_{a=1}^n \braket{z^\alpha_a|\ov{X}X^\transpose|z^\alpha_a}>0.
\end{equation}
Repeating the steps from the proof of  Proposition \ref{prop:no-go}, in particular \eqref{bip},  shows that the transformed polyhedron does not close as well.

\medskip

Finally, when the $\lambda_\alpha$ appear with multiplicities, we do not have a discrete family of polyhedra. The compact symplectic group symmetries do not leave invariant the observables $e_{ab}^{\alpha}, f_{ab}^{\alpha}, \tilde f_{ab}^{\alpha}$, though the total observables $e_{ab}, f_{ab}, \tilde f_{ab}$ are invariant by construction. Given a fixed value of the $|z^\alpha_a\ra$, we can reconstruct a polyhedron $\alpha$ with $N$ faces, which will have a global symmetry given by $\SU(2)$. Hence we can get in this way $k$ polyhedra with $N$ faces. However performing a transformation in  $\bigtimes_{i=1}^\ell \Sp(2\mu_i) / \SU(2)^{\times k}$ will give different values for the $|z^\alpha_a\ra$ (essentially mixing spinors associated to different polyhedra) and the new  polyhedra will be totally different than the previous ones when the multiplicity $\mu_i$ is higher than 1. Hence the semi-classical limit in this case can be seen as a set of  families of polyhedra given by the coset 
\be\bigtimes_{i=1}^\ell \Sp(2\mu_i) / \SU(2)^{\times k}.\ee
Note that some families can be only with a unique element, whereas some others can have a \textit{continuum} of polyhedra.  For example, let us consider an intertwiner with $N$ legs, such that  $N>7$, with $\zeta$ of rank 8 (hence $k=4$) and $\lambda_1=\lambda_2$,  so $\mu=2$ we have (recalling that $\Sp(2)=\SU(2)$)
$$
\left(\Sp(2\mu)\times\Sp(2)\times \Sp(2) \right) / \SU(2)^{\times 4}= \left(\Sp(4)\times\Sp(2)\times \Sp(2)\right)  / \SU(2)^{\times 4} = \Sp(4)  / \SU(2)^{\times 2}.
$$
Hence  we have two  $N$ faces polyhedra, with total area specified by $\lambda_3$ and  $\lambda_4$, and a continuous family of 2 polyhedra, with $N$ faces,  with each polyhedron a total area specified by $\lambda_1$.  

\smallskip

As a more explicit example let us consider now a 4 legs intertwiner, with $\zeta$ of rank 4 (hence $k=2$) and $\lambda_1=\lambda_2=\sqrt\demi$. We take explicitly for $\zeta$  
\be
\zeta=\demi U \begin{pmatrix}
\sigma & 0 \\ 0 & \sigma
\end{pmatrix}U^\transpose \textrm{ with } U=\demi \left(\begin{array}{cccc}
1&1& \sqrt{\demi}(-1-i)&\sqrt{\demi}(-1+i)\\
1&-1&\sqrt{\demi}(1-i)&\sqrt{\demi}(1+i)\\
i&1&0&\sqrt{2}\\
-i&1&\sqrt2&0
\end{array}\right)
\ee
From the unitary $U$, we can identify the spinors
 \bes
\ket{z_1^1}= \left(\begin{array}{c}\demi\\\demi \end{array}\right), \, \ket{z_2^1}= \left(\begin{array}{c}\demi\\-\demi \end{array}\right), \,  \ket{z_3^1}= \left(\begin{array}{c}\frac i2\\\frac i2 \end{array}\right), \,  \ket{z_4^1}= \left(\begin{array}{c}-\frac i2\\\frac i2 \end{array}\right),\\
\ket{z_1^2}= \frac{\sqrt{2}}{4}\left(\begin{array}{c}-1-i\\-1+i \end{array}\right), \, \ket{z_2^2}= \frac{\sqrt{2}}{4}\left(\begin{array}{c}1-i\\1+i \end{array}\right), \,  \ket{z_3^2}= \left(\begin{array}{c}0\\\sqrt{\demi} \end{array}\right), \,  \ket{z_4^2}= \left(\begin{array}{c}0\\\sqrt\demi \end{array}\right),
\ees
which in turn generate the normals that tell us about the two polyhedra geometries.
 \bes
\vec V^1_1= \left(\begin{array}{c}\frac14\\0\\0 \end{array}\right), \, \vec V^1_2= \left(\begin{array}{c}-\frac14\\0\\0 \end{array}\right), \,  \vec V^1_3= \left(\begin{array}{c}0\\-\frac14\\0 \end{array}\right), \, \vec V^1_4= \left(\begin{array}{c}0\\\frac14\\0 \end{array}\right),\\
\vec V^2_1= \left(\begin{array}{c}0\\-\frac14\\0 \end{array}\right), \, \vec V^2_2= \left(\begin{array}{c}0\\\frac14\\0 \end{array}\right), \,  \vec V^2_3= \left(\begin{array}{c}0\\0\\-\frac14 \end{array}\right), \,  \vec V^2_4= \left(\begin{array}{c}0\\0\\-\frac14\end{array}\right),
\ees
We note that the the two tetrahedra are actually both degenerated. Let us now perform a unitary transformation given by $U\dr UW$, with $W\in \Sp(4)$, which leaves $M$ invariant as it is easy to check.
\be
W=\sqrt\demi\left(\begin{array}{cc}
\one_2&\one_2\\
\one_2&-\one_2
\end{array}\right), \quad \dr \quad  
\ket{\tilde z^\pm_a}= \sqrt\demi\left( \ket{z^1_a}\pm   \ket{z^2_a}\right) 
\ee
Skipping the explicit expression of the spinors, we get the normals given by
 \bes
\vec V_1^\pm=\frac18 \left(\begin{array}{c}1\mp\sqrt2\\-1\pm\sqrt2\\0 \end{array}\right), \, \vec V_2^\pm= \frac18\left(\begin{array}{c}-1\\1\\\pm\sqrt2 \end{array}\right), \,  \vec V^\pm_3= \frac18\left(\begin{array}{c}0\\-1\mp\sqrt2\\-1\mp\sqrt2 \end{array}\right), \, \vec V^\pm_4= \frac18\left(\begin{array}{c}\pm\sqrt2\\1\\1 \end{array}\right).
\ees
We note now that the two polyhedra, indexed by $\pm$ are non-degenerated. The area of the faces changed, as one can easily check by evaluating the norm of the normals, but not   their respective total area. 

\smallskip

So to summarize, when the $\lambda_\alpha$ have some multiplicity, we cannot associate a discrete set of polyhedra to a single coherent state, instead we  have to consider a continuum of polyhedra. In  a sense the standard picture of the semi-classical limit breaks down.  

\smallskip

It is quite interesting that the degenerate structure appears as soon as the semi-classical polyhedra have the same \textit{total} area specified by the $\lambda_\alpha$. The fact that they might look very different, with individual faces of different values for the area does not affect the degeneracy.

\section*{Discussion}

We have identified  the algebra of observables for a $\SU(2)$ intertwiner to be $\so^*(2N)$ and constructed the associated coherent states. We studied the semi-classical limit of these coherent states which happened to be more subtle than the previously constructed coherent  states of a $\SU(2)$ intertwiner \cite{Livine:2011gp}. According to the nature of the  matrix parametrizing the coherent state, the semi-classical limit can give rise to families of convex polyhedra with $N$ faces that can be discrete (ie with a single polyhedron) or continuous (ie with an infinite number of polyhedra). The nature of the family, being discrete of continuous is characterized by the value of the total area of the polyhedron. If at least two polyhedra have the same total area, we will get a continuum of polyhedra with the same total area. 

\medskip

The $\SO^*(2N)$ states we have introduced have already been discovered in the literature \cite{Bonzom:2012bn,Freidel:2012ji,Freidel:2013fia,Bonzom:2015ova, Bianchi:2016hmk}, motivated by different reasons. However these states were always defined in terms of a matrix of rank 2, which is in the semi-classical limit the best scenario, since we only get one polyhedron in this case.   It would be interesting to see  
how the results of \cite{Bonzom:2012bn,Freidel:2012ji,Freidel:2013fia,Bonzom:2015ova,Bianchi:2016hmk} are affected when dealing with a more general coherent state. 

\medskip

There are some obvious generalizations to consider. 
 \smallskip

First, the algebra of observables has been deformed to the quantum group case, to deal with $\cU_q(\su(2))$ intertwiners \cite{Dupuis:2013lka}. Hence the deformed algebra $\cU_q(\so^*(2N))$ has been identified (not the co-algebra sector though). It would be interesting to see how the coherent states are generalized in this case. When $q$ is real, since the representation theory of $\cU_q(\su(2))$ is very similar in a sense to the one of $\cU(\su(2))$, we might not expect some great differences, either to define the coherent states or in term of the semi-classical limit. The case $q$ root of unity might be more interesting since the representation theory then changes drastically. This could affect somehow the semi-classical limit. This is to be explored.

\smallskip

Another interesting generalization would be the Lorentzian case. The spinorial formalism has been generalized to deal with $\SU(1,1)$ intertwiners \cite{Sellaroli:2014ega,Girelli:2015ija}. This formalism is more subtle than the Euclidian case, since when dealing with continuous representations, the "observables" acting on an intetwiner defined in terms of unitary irreps might not give an intertwiner defined in terms of unitary irreps. Nevertheless the algebra of observables can be defined. It would be interesting to see whether we can do some similar calculations as done here. We leave this for later investigations.

\section*{Acknowledgements}
 We would like to thank Etera Livine for some interesting discussions and comments on the paper as well as for pointing us out some  references where the $\SO^*(2N)$ coherent states  already appeared.  

\appendix

\section{Perelomov coherent states}\label{app:Perelomov}
Recall that generalised coherent states for a unitary irreducible module \( V \) of a generic Lie Group \( G \) are defined as
\begin{equation}
\ket{g}:=g\ket{\psi_0},\quad g\in G,
\end{equation}
where \( \ket{\psi_0}\in V \) is a fixed state of norm \( 1 \). Note that, at this stage, there is no guarantee that two coherent states labelled by different group elements indeed describe physically different states (i.e., they are not the same vector up to a phase\footnote{Note that since the representation is unitary and \( \ket{\psi_0} \) has norm \( 1 \), so does every \( \ket{g} \).}). In fact, let  \( H\subseteq G \) be the maximal subgroup that leaves \( \ket{\psi_0} \) invariant up to a phase, that is
\begin{equation}
h\ket{\psi_0}=e^{i\theta(h)}\ket{\psi_0},\quad \forall h\in H,
\end{equation}
which will be called the \textit{isotropy subgroup} for \( \ket{\psi_0} \): it is obvious that if \( g_2\in g_1 H \) then
\begin{equation}
\ket{g_2}=e^{i\theta} \ket{g_1},
\end{equation}
i.e., the two states are equivalent. The inequivalent coherent states are labelled by elements of the \textit{left coset space}
\begin{equation}
G/H:=\set{gH\setst g\in G},
\end{equation}
and are given by
\begin{equation}
\ket{x}:=\ket{g_x}=g_x\ket{\psi_0}, \quad \forall x\in G/H,
\end{equation}
where \( g_x\in x \) is a representative of the equivalence class \( x \).

\section{Proof of Proposition \ref*{prop:action of U(n) on highest weight}}\label{proof-prop1}
First recall that the exponential map of \( \GL(N,\C)\) is surjective, so that any element of \( \GL(N,\C) \) can be written in the form \( e^{E_\alpha}\equiv e^\alpha \), for some \( \alpha\in M_n(\C) \). Using the fact that
\begin{equation}
e^{-E_\alpha}\ket{0}=e^{-\tr(\alpha)}\ket{0}=\frac{1}{\det(e^\alpha)}\ket{0}
\end{equation}
and
\begin{equation}
\widetilde F_{12} \equiv \half \widetilde F_{\xi_0},
\end{equation}
we can write
\begin{equation}
e^{E_\alpha}\ket{\psi_J}= \frac{\det(e^\alpha)}{\sqrt{J!(J+1)!}}e^{E_\alpha}(\half \widetilde F_{\xi_0})^J e^{-E_\alpha}\ket{0} \equiv \frac{\det(e^\alpha)}{\sqrt{J!(J+1)!}}\big(\half e^{E_\alpha} \widetilde F_{\xi_0} e^{-E_\alpha}\big)^J\ket{0}.
\end{equation}
Using the well known formula
\begin{equation}
e^A B e^{-A} = B + [A,B] + \frac{1}{2!}[A,[A,B]] + \frac{1}{3!}[A,[A,[A,B]]] + \dotsb
\end{equation}
and the commutation relation
\begin{equation}
[E_\alpha,\widetilde F_z] = \widetilde F_{\alpha z + z \alpha^\transpose}
\end{equation}
we eventually find that
\begin{equation}
e^{E_\alpha}\ket{\psi_J} = \frac{\det(e^\alpha)}{\sqrt{J!(J+1)!}}\paren{\half \widetilde{F}_{e^\alpha\xi_0 e^{\alpha^\transpose}}}^J \ket{0}
\end{equation}
as expected.

\section{Proof of the Propositions of Section \ref*{sec:expectation} }\label{app:mat elements}
The easiest way to compute the matrix elements of the \( \so\Star(2N) \) generators \( E_{ab} \), \( F_{ab} \) and \( \widetilde{F}_{ab} \) in the coherent state basis is to make use of the \( 2N \) harmonic oscillator operators \( A_a \), \( B_a \); in particular, we are going to project the states \( \ket{\zeta} \) on the well-known harmonic oscillator coherent states.
Recall that \cite[chap.~3]{perelomov_book} coherent states for the representation of the Heisenberg group \( \mathrm{H}_{2N} \) with generators  satisfying
\begin{equation}
\bracks{A_a,A\Dagger_b}=\bracks{B_a,B\Dagger_b}=\delta_{ab}\1,
\end{equation}
acting on the vector space spanned by the vectors\footnote{Here we use the \textit{multi-index notation}, that is we have  \( \paren{A\Dagger}^{\mu}:= \paren{A\Dagger_1}^{\mu_1}\dotsb\paren{A\Dagger_n}^{\mu_N} \) and \( \mu!:=\mu_1!\dotsc\mu_N! \), with \( \mu\in \N_0^N \).}
\begin{equation}
\ket{\mu,\nu}=\frac{\paren{A\Dagger}^{\mu}}{\sqrt{\mu!}}\frac{\paren{B\Dagger}^{\nu}}{\sqrt{\nu!}}\ket{0},\quad \mu,\nu\in \N_0,
\end{equation}
where \( \ket{0}\equiv \ket{0,0} \) is the harmonic oscillator vacuum
\begin{equation}
A_a\ket{0}=B_a\ket{0}=0,
\end{equation}
are the vectors
\begin{equation}
\ket{\alpha,\beta}:=e^{-\half \paren{\alpha\Star\alpha +\beta\Star\beta} }\sum_{\mu,\nu\in\N_0^N}\frac{\alpha^\mu}{\sqrt{\mu!}}\frac{\beta^\nu}{\sqrt{\nu!}}\ket{\mu,\nu},\quad \alpha,\beta\in\C^N
\end{equation}
satisfying
\begin{equation}
A_a\ket{\alpha,\beta}=\alpha_a\ket{\alpha,\beta},\quad B_a\ket{\alpha,\beta}=\beta_a\ket{\alpha,\beta}.
\end{equation}
The resolution of the identity in terms of these coherent states is given by
\begin{equation}
\int_{\C^{2N}}\rd\mu(\alpha,\beta)\ket{\alpha,\beta}\bra{\alpha,\beta}=\1,
\end{equation}
where the measure of integration is\footnote{Here \( \Re \) and \( \Im \) denote respectively the real and imaginary part of a complex number.}
\begin{equation}
\rd\mu(\alpha,\beta)= \frac{1}{\pi^{2n}}\rd^n\Re(\alpha)\,\rd^n\,\Im(\alpha) \, \rd^n\Re(\beta)\,\rd^n\,\Im(\beta).
\end{equation}

We can now use the fact that
\begin{equation}
\begin{split}
\braket{\alpha,\beta|\zeta}&=\mathcal{N}(\zeta)\,\braket{\alpha,\beta|\exp\paren{\half \widetilde{F}_\zeta}|0}
=\mathcal{N}(\zeta)\,\braket{\alpha,\beta|0}\, e^{\beta\Star\zeta\ov{\alpha}}
=\mathcal{N}(\zeta)\, e^{\beta\Star\zeta\ov{\alpha}-\half\paren{\alpha\Star\alpha+\beta\Star\beta}}
\end{split}
\end{equation}
to write
\bes\label{eq:gaussian_int1}
\braket{\omega|\zeta}&=&\int_{\C^{2n}} \rd\mu(\alpha,\beta)\braket{\omega|\alpha,\beta}\braket{\alpha,\beta|\zeta}
= \mathcal{N}(\omega)\mathcal{N}(\zeta)
\int_{\C^{2n}} \rd\mu(\alpha,\beta)e^{\beta\Star\zeta\ov\alpha+\beta\Transpose \ov\omega\alpha - \alpha\Star\alpha-\beta\Star\beta}\nn\\
&=&\mathcal{N}(\omega)\mathcal{N}(\zeta)
\int_{\C^{2n}}\rd\mu(\alpha,\beta)\,\exp \bracks{
-\half
\left(\begin{array}{cccc} 
\alpha\Transpose & \beta\Transpose & \ov \alpha\Transpose & \ov \beta\Transpose
\end{array}\right)
\left(\begin{array}{cc|cc}
0 & \ov\omega & \1 & 0\\
-\ov\omega & 0 & 0 & \1\\\hline
\1 & 0 & 0 & \zeta \\
0 & \1 & -\zeta & 0
\end{array}\right)
\left(\begin{array}{c}
\alpha \\ \beta \\ \ov \alpha \\ \ov\beta
\end{array}\right)}
,
\ees
which is a Gaussian integral\footnote{In fact, we could also have calculated \( \braket{\omega|\zeta} \) by evaluating this integral.}. Although we already know the value of \( \braket{\omega|\zeta} \), we can use this expression to calculate the matrix elements of any operator built as a polynomial in the harmonic oscillator operators thanks to the following well-known proposition.
\begin{proposition}\label{prop:gaussian_integral}
Let
\begin{equation*}
S=\int e^{-\half x\Transpose A x} \,\rd^n x,
\end{equation*}
with \(A\in M_{n}(\C)\) symmetric and invertible, be a convergent Gaussian integral\footnote{It is assumed that all the requirements on $A$ such that the integral converges are satisfied.}. Then
\begin{equation*}
\int x_{a_1}x_{a_2}\dotsb x_{a_k} e^{-\half x\Transpose A x} \,\rd^n x
=
S
\left(\frac{\partial}{\partial J_{a_1}}\frac{\partial}{\partial J_{a_2}}\dotsb \frac{\partial}{\partial J_{a_k}} \left( e^{\half J\Transpose A^{-1} J}\right)\right)_{J=0}
\end{equation*}
for any \(k\in \N\); in particular, the integral vanishes whenever \(k\) is odd.
\end{proposition}

Proposition \ref{prop:gaussian_integral} can be used to find matrix elements by starting with \eqref{eq:gaussian_int1} and setting
\begin{equation}
J:=
\left(\begin{array}{c}
X \\ Y \\ \ov  X \\ \ov Y
\end{array}\right)
,\quad
A:=
\left(\begin{array}{cc|cc}
0 & \conj\omega & \1 & 0\\
-\conj\omega & 0 & 0 & \1\\\hline
\1 & 0 & 0 & \zeta \\
0 & \1 & -\zeta & 0
\end{array}\right) .
\end{equation}
One can easily check that
\begin{equation}
A^{-1}=
\left(\begin{array}{cc|cc}
0 & -\zeta(\1-\omega\Star\zeta)^{-1} & (\1-\zeta\omega\Star)^{-1} & 0 \\
\zeta(\1-\omega\Star\zeta)^{-1} & 0 & 0 & (\1-\zeta\omega\Star)^{-1}\\
\hline
(\1-\omega\Star\zeta)^{-1} & 0 & 0 & -(\1-\omega\Star\zeta)^{-1}\conj\omega\\
0 & (\1-\omega\Star\zeta)^{-1} & (\1-\omega\Star\zeta)^{-1}\conj\omega & 0
\end{array}\right),
\end{equation}
so that
\begin{equation}\label{eq:gaussian_source}
\begin{split}
S(\ov \omega,\zeta):=\half J\Transpose A^{-1} J =&\, Y\Transpose \zeta \paren{\1-\omega\Star\zeta}^{-1} X + \ov Y\Transpose \paren{\1-\omega\Star\zeta}^{-1}\ov\omega\ov X
\\
&+ \ov X\Transpose \paren{\1-\omega\Star\zeta}^{-1} X + \ov Y\Transpose \paren{\1-\omega\Star\zeta}^{-1} Y.
\end{split}
\end{equation}
Then, for any operator of the form\footnote{Here we use the multi-index notation again.}
\begin{equation}
p(A,B,A\Dagger,B\Dagger)=
A^{k_1} B^{k_2} \paren{A\Dagger}^{k_3} \paren{B\Dagger}^{k_4},
\quad k_1,k_2,k_3,k_4\in \N_0^N,
\end{equation}
where it is important that all the raising operators are on the right (anti-normal ordering)\footnote{If they are not, they can always be rewritten in this form up to some summands proportional to the identity, for which it is trivial to compute matrix elements.}, we have, as a consequence of Proposition \ref{prop:gaussian_integral},
\begin{equation}
\begin{split}
\braket{\omega|p(A,B,A\Dagger,B\Dagger)|\zeta} &= \int_{\C^{2N}}\rd\mu(\alpha,\beta)p(\alpha,\beta,\ov\alpha,\ov\beta)\braket{\omega|\alpha,\beta}\braket{\alpha,\beta|\zeta}
\\
&=\braket{\omega|\zeta}\left({p\paren{\nabla_X,\nabla_Y,\nabla_{\ov X},\nabla_{\ov Y}}}e^{S(\conj\omega,\zeta)}\right)_{X=Y=0} ,
\end{split}
\end{equation}
where
\begin{equation}
(\nabla_X)^k=\frac{\partial^{k_1}}{\partial X_1^{k_1}}\frac{\partial^{k_2}}{\partial X_2^{k_2}}\dotsb \frac{\partial ^{k_n}}{\partial X_n^{k_N}},\quad k\in\N^N.
\end{equation}

\subsection{Proof of Proposition \ref*{prop:coherent_matrix_elements}}
 First let us rewrite \( E_{ab} \) as
\begin{equation}
E_{ab}=A_b A\Dagger_a + B_b B\Dagger_a - \delta_{ab}
\end{equation}
using the commutation relations of the harmonic oscillators. Then we can insert the resolution of the identity for the harmonic oscillators coherent states to obtain
\begin{equation}
\begin{split}
\braket{\omega|A_b A\Dagger_a|\zeta} &= \int_{\C^{2N}}\rd\mu(\alpha,\beta)\braket{\omega|A_b|\alpha,\beta}\braket{\alpha,\beta|A\Dagger_a|\zeta}
\\
&= \int_{\C^{2N}}\rd\mu(\alpha,\beta)\,\ov{\alpha_a}\alpha_b\braket{\omega|\alpha,\beta}\braket{\alpha,\beta|\zeta}
\\
&= \mathcal{N}(\omega)\mathcal{N}(\zeta) \int_{\C^{2N}}\rd\mu(\alpha,\beta)\,\ov {\alpha_a}\alpha_b \,e^{\beta\Star\zeta\ov\alpha+\beta\Transpose \ov\omega\alpha - \alpha\Star\alpha-\beta\Star\beta};
\end{split}
\end{equation}
applying Proposition \ref{prop:gaussian_integral} together with  \eqref{eq:gaussian_int1} and\eqref{eq:gaussian_source}, we obtain
\begin{equation}
\begin{split}
\braket{\omega|A_b A\Dagger_a|\zeta} &= \braket{\omega|\zeta}\left({\frac{\partial}{\partial \ov X_a}\frac{\partial}{\partial X_b}} e^{\ov X\Transpose \paren{\1-\omega\Star\zeta}^{-1} X + \dotsb}\right)_{X=Y=0}
\\
&= \braket{\omega|\zeta}\bracks{\paren{\1-\omega\Star\zeta}^{-1}}_{ab}.
\end{split}
\end{equation}
Similarly
\begin{equation}
\braket{\omega|B_b B\Dagger_a|\zeta} = \braket{\omega|\zeta}\bracks{\paren{\1-\omega\Star\zeta}^{-1}}_{ab}
\end{equation}
so that
\begin{equation}
\braket{\omega|E_{ab}|\zeta}=\braket{\omega|\zeta}\bracks{2\paren{\1-\omega\Star\zeta}^{-1} -\1}_{ab}=\braket{\omega|\zeta}\bracks{\1 + 2\omega\Star\zeta\paren{\1-\omega\Star\zeta -\1}^{-1}}_{ab}
\end{equation}
as
\begin{equation}
\paren{\1-X}^{-1}=\1+ X\paren{\1-X}^{-1}.
\end{equation}
To obtain the matrix elements of \( F_{ab} \) we insert the resolution of the identity again, which gives
\begin{equation}
\begin{split}
\braket{\omega|B_a A_b|\zeta} &= \int_{\C^{2N}}\rd\mu(\alpha,\beta)\braket{\omega|B_a A_b|\alpha,\beta}\braket{\alpha,\beta|\zeta}
= \int_{\C^{2N}}\rd\mu(\alpha,\beta)\,\beta_a\alpha_b\braket{\omega|\alpha,\beta}\braket{\alpha,\beta|\zeta}
\\
&= \braket{\omega|\zeta} \left({\frac{\partial}{\partial Y_a}\frac{\partial}{\partial X_b}} e^{Y\Transpose \zeta\paren{\1-\omega\Star\zeta}^{-1} X + \dotsb}\right)_{X=Y=0}
= \braket{\omega|\zeta} \bracks{\zeta\paren{\1-\omega\Star\zeta}^{-1}}_{ab},
\end{split}
\end{equation}
leading to
\begin{equation}
\braket{\omega|F_{ab}|\zeta}=\braket{\omega|\zeta} \bracks{2\zeta\paren{\1-\omega\Star\zeta}^{-1}}_{ab}
\end{equation}
as
\begin{equation}
\bracks{\zeta\paren{\1-\omega\Star\zeta}^{-1}}\Transpose=-\paren{\1-\zeta\omega\Star}^{-1}\zeta=-\zeta\paren{\1-\omega\Star\zeta}^{-1}.
\end{equation}
 The matrix elements of \( \widetilde{F}_{ab} \) are easily obtained from the \( F_{ab} \) ones as
\begin{equation}
\begin{split}
\braket{\omega|\widetilde{F}_{ab}|\zeta} &= \ov{\braket{\zeta|F_{ab}|\omega}}
= \ov{\braket{\zeta|\omega}} \bracks{2\ov\omega\paren{\1-\zeta\omega\Star}^{-1}}_{ab}
= \braket{\omega|\zeta} \bracks{2\paren{\1-\omega\Star\zeta}^{-1}\ov\omega}_{ab}.
\end{split}
\end{equation}

\subsection{Proof of Proposition \ref*{prop:expected_values}}

The form of the expected values follows directly from Proposition \ref{prop:coherent_matrix_elements}. In order to calculate the variances, we will need the covariance\footnote{Note that the \( \mathcal{A}_a \) all commute, so there is no ordering ambiguity.}
\begin{equation}
\Cov(\mathcal{A}_a,\mathcal{A}_b):=\braket{\mathcal{A}_a\mathcal{A}_b}-\braket{\mathcal{A}_a}\braket{\mathcal{A}_b}.
\end{equation}
First note that
\begin{equation}
\begin{split}
4\mathcal{A}_a\mathcal{A}_b =&\, \paren{A_a A\Dagger_a + B_a B\Dagger_a-2}\paren{A_b A\Dagger_b + B_b B\Dagger_b-2}
\\
=& \, A_a A\Dagger_a A_b A\Dagger_b + B_a B\Dagger_a B_b B\Dagger_b + A_a A\Dagger_a B_b B\Dagger_b + B_a B\Dagger_a A_b A\Dagger_b
-4\mathcal{A}_a -4\mathcal{A}_b -4
\\
=&\, A_a A_b A\Dagger_a A\Dagger_b + B_a B_b B\Dagger_a B\Dagger_b + A_a B_b A\Dagger_a B\Dagger_b + A_b B_a A\Dagger_b B\Dagger_a
-4\mathcal{A}_a -4\mathcal{A}_b -2\delta_{ab}\mathcal{A}_a - 4 -2\delta_{ab}.
\end{split}
\end{equation}
Making use of the resolution of the identity for the \( \mathrm{H}_{2n} \) coherent states we get\footnote{To simplify notation we define \( \sigma := \paren{\1-\zeta\Star\zeta}^{-1} \).}
\begin{equation}
\begin{split}
\braket{A_a A_b A\Dagger_a A\Dagger_b} &= \int_{\C^{2N}}\rd\mu(\alpha,\beta)\braket{\zeta|A_a A_b|\alpha,\beta}\braket{\alpha,\beta| A\Dagger_a A\Dagger_b|\zeta}
= \int_{\C^{2N}}\rd\mu(\alpha,\beta)\,\alpha_a \alpha_b \ov\alpha_a \ov\alpha_b \braket{\zeta|\alpha,\beta}\braket{\alpha,\beta|\zeta}
\\
&= \left({\frac{\partial}{\partial X_a}\frac{\partial}{\partial X_b}\frac{\partial}{\partial \ov X_a}\frac{\partial}{\partial \ov X_b}}e^{\ov{X}\Transpose\sigma X+\dotsb}\right)_{X=Y=0}
= \frac{\partial}{\partial X_a}\frac{\partial}{\partial X_b} \paren{(\sigma X)_a + (\sigma X)_b}
=\sigma_{aa}\sigma_{bb}+\sigma_{ab}\sigma_{ba}
\end{split}
\end{equation}
and similarly
\begin{equation}
\braket{B_a B_b B\Dagger_a B\Dagger_b} = \sigma_{aa}\sigma_{bb}+\sigma_{ab}\sigma_{ba},
\end{equation}
while for the term with both harmonic oscillators we have
\begin{equation}
\begin{split}
\braket{A_a B_b A\Dagger_a B\Dagger_b} &= \int_{\C^{2N}}\rd\mu(\alpha,\beta)\braket{\zeta|A_a B_b|\alpha,\beta}\braket{\alpha,\beta| A\Dagger_a B\Dagger_b|\zeta}
= \int_{\C^{2N}}\rd\mu(\alpha,\beta)\,\alpha_a \beta_b \ov\alpha_a \ov\beta_b \braket{\zeta|\alpha,\beta}\braket{\alpha,\beta|\zeta}
\\
&= \left({\frac{\partial}{X_a}\frac{\partial}{\partial Y_b}\frac{\partial}{\partial \ov X_a}\frac{\partial}{\partial \ov Y_b}}e^{Y\Transpose \zeta \sigma X + \ov Y\Transpose \sigma \ov\zeta \ov X +\ov{X}\Transpose\sigma X+ \ov Y \sigma Y}\right)_{X=Y=0}
\\
&= \left({\frac{\partial}{\partial X_a}\frac{\partial}{\partial Y_b}\frac{\partial}{\partial \ov X_a}} \paren{(\sigma Y)_b + \paren{\sigma \ov\zeta \ov X}_b}
\, e^{Y\Transpose \zeta \sigma X +\ov{X}\Transpose\sigma X+ \dotsb}\right)_{X=Y=0}
\\
&= \left({\frac{\partial}{\partial X_a}\frac{\partial}{\partial Y_b}} \paren{(\sigma X)_a (\sigma Y)_b + \paren{\sigma \ov\zeta}_{ba}}
\, e^{Y\Transpose \zeta \sigma X + \dotsb}\right)_{X=Y=0}
= \sigma_{aa}\sigma_{bb}+\paren{\sigma \ov\zeta}_{ba} \paren{\zeta\sigma}_{ba}
\\
&= \sigma_{aa}\sigma_{bb} + \paren{\sigma\zeta\Star}_{ab}\paren{\zeta\sigma}_{ba}
= \sigma_{aa}\sigma_{bb} + \paren{\sigma\zeta\Star}_{ba}\paren{\zeta\sigma}_{ab}
\end{split}
\end{equation}
Eventually we can compute the covariance as\footnote{Recall that \( \zeta\Star\zeta\sigma=\sigma-\1 \), so that \( \braket{\mathcal{A}_a}=\sigma_{aa}-1 \).}
\begin{equation}
\begin{split}
\Cov(\mathcal{A}_a,\mathcal{A}_b) =& \sigma_{aa}\sigma_{bb}+\half\sigma_{ab}\sigma_{ba} + \half \paren{\sigma\zeta\Star}_{ab}\paren{\zeta\sigma}_{ba} -\sigma_{aa} - \sigma_{bb} 
\\
& - \half \delta_{ab}\sigma_{ab} + 1 - \sigma_{aa}\sigma_{bb}+\sigma_{aa}+\sigma_{bb}-1
\\
=& \half\sigma_{ab}\sigma_{ba} + \half \paren{\sigma\zeta\Star}_{ab}\paren{\zeta\sigma}_{ba} - \half \delta_{ab}\sigma_{ab},
\end{split}
\end{equation}
which leads to\footnote{Note that \( \zeta\sigma \) is antisymmetric.}
\begin{equation}
\Var(\mathcal{A}_a) := \Cov(\mathcal{A}_a,\mathcal{A}_a)=\half\sigma_{aa}(\sigma_{aa}-1)\equiv
\half\braket{\mathcal{A}_a}\paren{\braket{\mathcal{A}_a}+1}
\end{equation}
and
\begin{equation}
\Var(\mathcal{A}) :=\sum_{a,b}\Cov(\mathcal{A}_a,\mathcal{A}_b)=\tr(\sigma^2-\sigma) \equiv
\sum_{a,b}\braket{\mathcal{A}_{ab}}\paren{\braket{\mathcal{A}_{ab}}+\delta_{ab}}.
\end{equation}

The coefficient of variation for the total area is then given by
\begin{equation}
\frac{\sqrt{\Var{\mathcal{A}}}}{\braket{\mathcal{A}}}=\frac{\sqrt{\tr\bracks{\sigma(\sigma-\1)}}}{\tr(\sigma-\1)}\geq 0;
\end{equation}
making use of the fact that, as both \( \sigma \) and \( \sigma-\1 \) are positive semi-definite\footnote{Note that as \( (\1-\zeta\Star\zeta) \leq \1\), we must have  \( \sigma\geq \1 \).},
\begin{equation}
\tr\bracks{\sigma(\sigma-1)}\leq\tr(\sigma)\tr(\sigma-\1),
\end{equation}
we obtain an upper bound for the coefficient of variation,
\begin{equation}
\frac{\sqrt{\Var{\mathcal{A}}}}{\braket{\mathcal{A}}} \leq \paren{\frac{\tr(\sigma)}{\tr(\sigma)-N}}^{\half}.
\end{equation}
When the non-zero eigenvalues of \( \zeta\Star\zeta\) approach \( \1 \) we have \( \tr(\sigma)\rightarrow\infty \), so that
\begin{equation}
\frac{\sqrt{\Var{\mathcal{A}}}}{\braket{\mathcal{A}}} \lesssim 1 \quad \mbox{when} \quad \tr(\sigma)\rightarrow \infty,
\end{equation}
as expected.

\subsection{Proof of Proposition \ref*{prop:probability}}\label{proof:probability}

Let us define
\begin{equation}
\ket{J,\zeta}:=\left(\half \widetilde F_\zeta\right)^J\ket{0},\quad J\in \N,
\end{equation}
which  are eigenvectors of $\mathcal{A}$, with
\begin{equation}
\mathcal{A}\ket{J,\zeta}=J\ket{J,\zeta},
\end{equation}
as $\widetilde F_\zeta$ adds one quantum of area each time\footnote{In fact \( [\mathcal{A},\widetilde F_{\zeta}]=\widetilde F_{\zeta}\).}. 
The $\SO^*(2N)$ coherent states can then be written as
\begin{equation}
\begin{split}
\ket{\zeta}&=\det(\1-\zeta\Star\zeta)^{\half} \exp\paren{\half \widetilde F_\zeta} \ket{0}
=\det(\1-\zeta\Star\zeta)^{\half} \sum_{J=0}^\infty\frac{1}{J!}\ket{J,\zeta}.
\end{split}
\end{equation}
Since the $\ket{J,\zeta}$ states are mutually orthogonal\footnote{As they are eigenvectors of a self-adjoint operator, with different eigenvalues.}, the probability that $\ket{\zeta}$ is measured with total area $J$ is given by
\begin{equation}\label{eq:probability_from_projection}
P_\zeta(J)\equiv\frac{\abs{\braket{J,\zeta|\zeta}}^2}{\braket{J,\zeta|J,\zeta}}
=\frac{\det(\1-\zeta\Star\zeta)}{(J!)^2}\braket{J,\zeta|J,\zeta};
\end{equation}
it remains to calculate the norm squared of the state $\ket{J,\zeta}$.
Recall that
\begin{equation}
\bracks{\half F_\zeta,\half \widetilde  F_\zeta}=E_{\frac{1}{4}(\zeta-\zeta\Transpose)(\zeta-\zeta\Transpose)\Star}= E_{\zeta\zeta\Star}
\end{equation}
and
\begin{equation}
\bracks{E_{\zeta\zeta\Star},\half \widetilde  F_\zeta}=\half \widetilde F_{\zeta\zeta\Star\zeta+\zeta\ov{\zeta}\zeta\Transpose}=\widetilde F_{\zeta\zeta\Star\zeta};
\end{equation}
moreover, since $\zeta$ is of rank $2$, one has 
\begin{equation}
\zeta\zeta\Star\zeta=\half \tr(\zeta\zeta\Star)\zeta,
\end{equation}
so that
\begin{equation}
\begin{split}
\bracks{E_{\zeta\zeta\Star},\paren{\half \widetilde  F_\zeta}^k}&=\sum_{\ell=1}^{k}\paren{\half \widetilde  F_\zeta}^{\ell-1}
\,\bracks{E_{\zeta\zeta\Star},\half \widetilde  F_\zeta}\,
\paren{\half \widetilde  F_\zeta}^{k-\ell}
\\
&=\sum_{\ell=1}^{k}\paren{\half \widetilde  F_\zeta}^{\ell-1}
\,\widetilde F_{\zeta\zeta\Star\zeta}\,
\paren{\half \widetilde  F_\zeta}^{k-\ell}
\\
&=k\tr(\zeta\zeta\Star)\paren{\half \widetilde  F_\zeta}^{k}.
\end{split}
\end{equation}
It follows that\footnote{Recall that $F_\alpha \ket{0}=0$ and that $E_{\alpha}\ket{0}=\tr(\alpha)\ket{0}$.}
\begin{equation}
\begin{split}
\half F_\zeta\,\paren{\half \widetilde  F_\zeta}^{J}\ket{0}&=\bracks{\half  F_\zeta,\paren{\half \widetilde  F_\zeta}^{J}}\ket{0}
\\
&=\sum_{k=1}^J \paren{\half \widetilde  F_\zeta}^{k-1}
\,\bracks{\half F_\zeta,\half \widetilde  F_\zeta}\,
\paren{\half \widetilde  F_\zeta}^{J-k}\ket{0}
\\
&=\sum_{k=1}^J \paren{\half \widetilde  F_\zeta}^{k-1}
\,E_{\zeta\zeta\Star}\,
\paren{\half \widetilde  F_\zeta}^{J-k}\ket{0}
\\
&=J\tr(\zeta\zeta\Star)\paren{\half \widetilde  F_\zeta}^{J-1}\ket{0}+
\sum_{k=1}^J \paren{\half \widetilde  F_\zeta}^{k-1}
\,\bracks{E_{\zeta\zeta\Star},\paren{\half \widetilde  F_\zeta}^{J-k}}\ket{0}
\\
&=J(J+1)\half\tr(\zeta\Star\zeta)\paren{\half \widetilde  F_\zeta}^{J-1}\ket{0}
\end{split}
\end{equation}
and in particular
\begin{equation}
\begin{split}
\braket{J,\zeta|J,\zeta}&=\braket{0|\paren{\half  F_\zeta}^{J}\paren{\half \widetilde  F_\zeta}^{J}|0}\\
&=J(J+1)\tr(\zeta\Star\zeta)\braket{0|\paren{\half  F_\zeta}^{J-1}\paren{\half \widetilde  F_\zeta}^{J-1}|0}
\\
&=J(J+1)\half\tr(\zeta\Star\zeta)\braket{J-1,\zeta|J-1,\zeta}.
\end{split}
\end{equation}
Solving the recurrence relation with $\braket{0,\zeta|0,\zeta}=\braket{0|0}=1$ we obtain
\begin{equation}
\braket{J,\zeta|J,\zeta}=J!(J+1)!\paren{\half\tr(\zeta\Star\zeta)}^J
\end{equation}
which, plugged in \eqref{eq:probability_from_projection}, gives
\begin{equation}
P_\zeta(J)=\det(\1-\zeta\Star\zeta)\paren{\half\tr(\zeta\Star\zeta)}^J (J+1)
\end{equation}
as expected.

\section{Proof of proposition \ref*{prop:symmetries}}\label{proof of prop sym}
For any anti-symmetric matrix \( \zeta\in M_N(\C) \) of rank \( 2k \) there is \( U\in \mathrm{U}(N) \) such that
\begin{equation}
\zeta = U M U^\transpose,
\end{equation}
where
\begin{equation}
M = \bigoplus_{\alpha=1}^k \sigma_\alpha \oplus 0_{N-2k},\quad
\sigma_\alpha = \lambda_\alpha
\begin{pmatrix}
0 & -1 \\ 1 & 0
\end{pmatrix}
\end{equation}
and
\begin{equation}
\lambda_1 \geq \lambda_2\geq \dotsb\geq \lambda_k>0
\end{equation}
are the positive square roots of the eigenvalues of \( \zeta^*\zeta \). The unitary matrix \( U \) is not unique, as the matrix \( \widetilde U= UW \), \( W\in\mathrm{U}(N) \), satisfies
\begin{equation}
\zeta = \widetilde{U} M \widetilde{U}^\transpose
\end{equation}
whenever
\begin{equation}\label{eq:V}
WMW^\transpose = M.
\end{equation}
Let us find the generic form of the matrices \( W\in\mathrm{U}(N) \) satisfying \eqref{eq:V}. Let
\begin{equation}
W =
\begin{pmatrix}
A_{11} & A_{12} & \cdots & A_{1k} & B_1\\
A_{21} & A_{22} & \cdots & A_{2k} & B_2\\
\vdots & \vdots & \ddots & \vdots & \vdots\\
A_{k1} & A_{k2} & \cdots & A_{kk} & B_k\\
C_1    & C_2    & \cdots & C_k    & D\\
\end{pmatrix},
\end{equation}
with \( A_{\alpha\beta}\in M_2(\C)\), \(B_\alpha\in M_{2,N-2k}(\C) \), \(C_\alpha\in M_{N-2k,2}(\C) \) and \(D\in M_{N-2k}(\C) \). Then \( V \) satisfies \eqref{eq:V} if and only if
\begin{subequations}
\begin{gather}
\label{eq:V_AA}
\sum_{\gamma=1}^k A_{\alpha\gamma} \sigma_\gamma A_{\beta\gamma}^\transpose = \delta_{\alpha\beta} \sigma_\alpha
\\\label{eq:V_AC-CC}
\sum_{\beta=1}^k A_{\alpha\beta} \sigma_\beta C_\beta^\transpose =
\sum_{\alpha=1}^k C_\alpha \sigma_\alpha C_\alpha^\transpose  = 0.
\end{gather}
\end{subequations}
Moreover, since \( W \) is unitary it must satisfy the condition \( W^* W = \1 \), which is equivalent to
\begin{subequations}
\begin{gather}
\label{eq:unitary_AA}
\sum_{\gamma=1}^k A_{\gamma\alpha}^* A_{\gamma\beta} + C_\alpha^* C_\beta = \delta_{\alpha\beta} \1_2
\\\label{eq:unitary_AB}
\sum_{\beta=1}^k A_{\beta\alpha}^* B_\beta + C_\alpha^* D = 0
\\\label{eq:unitary_BB}
\sum_{\alpha=1}^k B_\alpha^* B_\alpha + D^* D  = \1_{N-2k}.
\end{gather}
\end{subequations}
Putting together \eqref{eq:unitary_AA} and \eqref{eq:V_AC-CC} we get
\begin{equation}
\sum_{\gamma,\varepsilon} A^*_{\varepsilon\beta} A_{\varepsilon\gamma} \sigma_\gamma A_{\alpha\gamma}^\transpose = \sum_{\gamma}\paren{\delta_{\beta\gamma} \sigma_\gamma A^\transpose_{\alpha\gamma} - C^*_\beta C_\gamma \sigma_\gamma A^\transpose_{\alpha\gamma} } = \sigma_\beta A^\transpose_{\alpha\beta},
\end{equation}
while we know from \eqref{eq:V_AA} that
\begin{equation}
\sum_{\gamma,\varepsilon} A^*_{\varepsilon\beta} A_{\varepsilon\gamma} \sigma_\gamma A_{\alpha\gamma}^\transpose = \sum_{\varepsilon} \delta_{\varepsilon\alpha} A^*_{\varepsilon\beta} \sigma_\varepsilon = A^*_{\alpha\beta} \sigma_\alpha.
\end{equation}
It follows that
\begin{equation}\label{eq:Az_zAbar}
A_{\alpha\beta} \sigma_\beta = \sigma_\alpha \ov A_{\alpha\beta}.
\end{equation}
Similarly, we can show that
\begin{equation}
0=\sum_{\beta,\varepsilon} A^*_{\varepsilon\alpha} A_{\varepsilon\beta} \sigma_\beta C_{\beta}^\transpose = \sum_{\beta} \paren{\delta_{\alpha\beta} \sigma_\beta C^\transpose_\beta - C^*_{\alpha}C_\beta \sigma_\beta C^\transpose_\beta } = \sigma_\alpha C^\transpose_\alpha,
\end{equation}
which means that, as each \( \sigma_\alpha \) is invertible, that
\begin{equation}\label{eq:C=0}
C_\alpha = 0.
\end{equation}
Since \( W \) is unitary, it also satisfies \( WW^* = \1 \), which in particular implies that
\begin{equation}
\sum_{\alpha} C_\alpha C^*_\alpha + D D^* \equiv D D^* = \1_{N-2k},
\end{equation}
i.e., \( D\in \mathrm{U}(N-2k) \), from which it follows that \( D^* D = \1_{N-2k} \) as well. Plugging in this result in \eqref{eq:unitary_BB} we see that
\begin{equation}
\sum_{\alpha} B^*_\alpha B_\alpha = 0.
\end{equation}
Since each \( B^*_\alpha B_\alpha \) is positive semi-definite we must have  $B_\alpha =0$.
Now, using the fact that for any \( X\in M_2(\C) \)
\begin{equation}
X
\begin{pmatrix}
0 & -1 \\ 1 & 0
\end{pmatrix}
X^\transpose = \det(X)
\begin{pmatrix}
0 & -1 \\ 1 & 0
\end{pmatrix},
\end{equation}
we see from \eqref{eq:Az_zAbar} that
\begin{subequations}
\begin{align}
A^*_{\alpha\beta} A_{\alpha\beta} \sigma_\beta &= \ov{a}_{\alpha\beta} \sigma_\alpha\\
\sigma_\alpha \ov{A}_{\alpha\beta} A^\transpose_{\alpha\beta} &= a_{\alpha\beta} \sigma_\beta
\end{align}
\end{subequations}
where \(  a_{\alpha\beta} := \det(A_{\alpha\beta}) \), that is
\begin{subequations}\label{eq:A*A_AA*}
\begin{align}
A^*_{\alpha\beta} A_{\alpha\beta} &= \frac{\lambda_\alpha}{\lambda_\beta} \ov{a}_{\alpha\beta} \1_2 \equiv \frac{\lambda_\alpha}{\lambda_\beta} {a}_{\alpha\beta} \1_2
\\
A_{\alpha\beta} A^*_{\alpha\beta} &= \frac{\lambda_\beta}{\lambda_\alpha} {a}_{\alpha\beta} \1_2,
\end{align}
\end{subequations}
as we must have  \( a_{\alpha\beta} \in \R \).

Eqs. \eqref{eq:A*A_AA*} have two important consequences. First, notice that
\begin{equation}
\det(A_{\alpha\beta})=0
\quad \Rightarrow \quad
A^*_{\alpha\beta} A_{\alpha\beta}=0
\quad \Rightarrow \quad
A_{\alpha\beta}=0,
\end{equation}
that is each \( A_{\alpha\beta} \) is either invertible or zero. Secondly, taking the determinant on both sides of the two equations in \eqref{eq:A*A_AA*}, we get
\begin{equation}
a_{\alpha\beta}^2 = \frac{\lambda_\alpha^2}{\lambda_\beta^2} a^2_{\alpha\beta} = \frac{\lambda_\beta^2}{\lambda_\alpha^2} a_{\alpha\beta}^2,
\end{equation}
which means that, \textit{unless} \( \lambda_\alpha = \lambda_\beta \), we must have  \( a_{\alpha\beta} = 0 \), which as we have seen implies \( A_{\alpha\beta}=0 \).

\subsection*{{Case 1: all \( \lambda_\alpha \) are distinct}}

If all the \( \lambda_\alpha \) are different from each other, we must have 
\begin{equation}
A_{\alpha\beta} = 0,\quad \alpha\neq \beta.
\end{equation}
Then it follows from \eqref{eq:unitary_AA} that
\begin{equation}
A^*_{\alpha\alpha} A_{\alpha\alpha} = \1_2,
\end{equation}
i.e., \( A_{\alpha\alpha} \in \mathrm{U}(2) \), and from \eqref{eq:V_AA} that
\begin{equation}
a_{\alpha\alpha} \sigma_\alpha = A_{\alpha\alpha} \sigma_\alpha A_{\alpha\alpha}^\transpose = \sigma_\alpha,
\end{equation}
so \( A_{\alpha\alpha} \in \SU(2) \). Thus the generic form of \( V \) is
\begin{equation}
W=
\begin{pmatrix}
W_1 & 0   & \cdots & 0 & 0\\
0 & W_2 & \cdots & 0 & 0\\
\vdots & \vdots & \ddots & \vdots & \vdots\\
0 & 0 & \cdots & W_k & 0\\
0 & 0 & \cdots & 0    & D\\
\end{pmatrix},
\end{equation}
with \( W_\alpha \in \SU(2)\) and \( D\in \mathrm{U}(N-2k) \).

\subsection*{{Case 2: some \( \lambda_\alpha \) are identical}}

When some of the \( \lambda_\alpha \) are the same, \( M \) has some additional invariance. Suppose there are \( \ell\leq k \) distinct \( \lambda_\alpha \), and let us denote them by
\begin{equation}
\Lambda_1\geq \Lambda_2 \geq\dotsb \geq \Lambda_\ell >0.
\end{equation}
Moreover, let \( \mu_i \) be the multiplicity of \( \Lambda_i \). As we have seen, \( A_{\alpha\beta}= 0 \)  when \( \lambda_\alpha\neq \lambda_\beta \), so that \( W \) must be of the form
\begin{equation}
W=
\begin{pmatrix}
W_1 & 0   & \cdots & 0 & 0\\
0 & W_2 & \cdots & 0 & 0\\
\vdots & \vdots & \ddots & \vdots & \vdots\\
0 & 0 & \cdots & W_\ell & 0\\
0 & 0 & \cdots & 0    & D\\
\end{pmatrix},
\quad W_i \in M_{2\mu_i}(\C),
\end{equation}
while
\begin{equation}
M=
\begin{pmatrix}
M_1 & 0   & \cdots & 0 & 0\\
0 & M_2 & \cdots & 0 & 0\\
\vdots & \vdots & \ddots & \vdots & \vdots\\
0 & 0 & \cdots & M_\ell & 0\\
0 & 0 & \cdots & 0    & 0\\
\end{pmatrix},
\quad M_i \in M_{2\mu_i}(\C),
\end{equation}
with
\begin{equation}
M_i = \Lambda_i \Omega_i,\quad \Omega_i=
\begin{pmatrix}
\sigma & 0   & \cdots & 0 \\
0 & \sigma & \cdots & 0 \\
\vdots & \vdots & \ddots & \vdots \\
0 & 0 & \cdots & \sigma \\
\end{pmatrix}, \quad \sigma=\begin{pmatrix}
0 & -1 \\ 1 & 0
\end{pmatrix}.
\end{equation}
It follows that \( W M W^\transpose = M \) if and only if, for each \( i=1,\dotsc,\ell \),
\begin{equation}
W_i \Omega_i W_i^\transpose = \Omega_i,
\end{equation}
that is \( W_i\in \Sp(2\mu_i,\C) \), the complex symplectic group. Since each \( W_i \) has to be unitary as well, we conclude that \( W \) leaves \( M \) invariant if and only if \( W_i \) belongs to the \textit{compact symplectic group}, i.e.,
\begin{equation}
W_i\in \Sp(2\mu_i) := \Sp(2\mu_i,\C)\cap \mathrm{U}(2\mu_i).
\end{equation}
Note that if \( \mu_i=1 \), then
\begin{equation}
W_i\in\Sp(2)\equiv \SU(2),
\end{equation}
as expected.

\bibliography{so*(2n)}

\end{document}